\documentclass[aps,prl,twocolumn,amsmath,amssymb,nofootinbib,superscriptaddress]{revtex4-1}
\usepackage{adjustbox}
\usepackage[utf8]{inputenc} 
\usepackage{hyperref}       
\usepackage{url}            
\usepackage{booktabs}       
\usepackage{amsfonts}       
\usepackage{nicefrac}       
\usepackage{microtype}      
\usepackage{natbib}

\usepackage[normalem]{ulem}
\usepackage{subcaption}
\usepackage{bm}
\usepackage{braket}
\usepackage{graphicx}
\usepackage{amsthm}
\usepackage{algorithmic}
\usepackage[linesnumbered,ruled]{algorithm2e}
\SetKwInOut{Parameter}{parameter}

\usepackage{mathtools}
\usepackage{nccmath}
\usepackage{xcolor}
\usepackage{caption}
\usepackage{titlesec}

\usepackage{soul}
\newcommand{\Pro}{\text{P}}

\newcommand{\rrr}[1]{\textcolor{black}{{{#1}}}}
\titlespacing{\title}{0pc}{0.1pc}{0.3pc}
\titlespacing{\section}{0pc}{0.1pc}{0.3pc}
\begin{document}

\title{Experimental Quantum Generative Adversarial Networks for Image Generation}
	
\author{He-Liang Huang}
\thanks{These two authors contributed equally}
\affiliation{Hefei National Laboratory for Physical Sciences at the Microscale and Department of Modern Physics, University of Science and Technology of China, Hefei 230026, China}
\affiliation{Shanghai Branch, CAS Center for Excellence in Quantum Information and Quantum Physics, University of Science and Technology of China, Shanghai 201315, China}
\affiliation{Shanghai Research Center for Quantum Sciences, Shanghai 201315, China}
\affiliation{Henan Key Laboratory of Quantum Information and Cryptography, Zhengzhou, Henan 450000, China}
\author{Yuxuan Du}
\thanks{These two authors contributed equally}
\affiliation{School of Computer Science, Faculty of Engineering, University of Sydney, Australia}
\author{Ming Gong}
\author{Youwei Zhao}
\author{Yulin Wu}
\affiliation{Hefei National Laboratory for Physical Sciences at the Microscale and Department of Modern Physics, University of Science and Technology of China, Hefei 230026, China}
\affiliation{Shanghai Branch, CAS Center for Excellence in Quantum Information and Quantum Physics, University of Science and Technology of China, Shanghai 201315, China}
\affiliation{Shanghai Research Center for Quantum Sciences, Shanghai 201315, China}
\author{Chaoyue Wang}
\affiliation{School of Computer Science, Faculty of Engineering, University of Sydney, Australia}
\author{Shaowei Li}
\author{Futian Liang}
\author{Jin Lin}
\author{Yu Xu}
\author{Rui Yang}
\affiliation{Hefei National Laboratory for Physical Sciences at the Microscale and Department of Modern Physics, University of Science and Technology of China, Hefei 230026, China}
\affiliation{Shanghai Branch, CAS Center for Excellence in Quantum Information and Quantum Physics, University of Science and Technology of China, Shanghai 201315, China}
\affiliation{Shanghai Research Center for Quantum Sciences, Shanghai 201315, China}
\author{Tongliang Liu}
\affiliation{School of Computer Science, Faculty of Engineering, University of Sydney, Australia}
\author{Min-Hsiu Hsieh}
\affiliation{Hon Hai Research Institute, Taipei 114, Taiwan}
\author{Hui Deng}
\author{Hao Rong}
\author{Cheng-Zhi Peng}
\author{Chao-Yang Lu}
\author{Yu-Ao Chen}
\affiliation{Hefei National Laboratory for Physical Sciences at the Microscale and Department of Modern Physics, University of Science and Technology of China, Hefei 230026, China}
\affiliation{Shanghai Branch, CAS Center for Excellence in Quantum Information and Quantum Physics, University of Science and Technology of China, Shanghai 201315, China}
\affiliation{Shanghai Research Center for Quantum Sciences, Shanghai 201315, China}
\author{Dacheng Tao}
\email{dacheng.tao@sydney.edu.au}
\affiliation{School of Computer Science, Faculty of Engineering, University of Sydney, Australia}
\author{Xiaobo Zhu}
\email{xbzhu16@ustc.edu.cn}
\author{Jian-Wei Pan}
\email{pan@ustc.edu.cn}
\affiliation{Hefei National Laboratory for Physical Sciences at the Microscale and Department of Modern Physics, University of Science and Technology of China, Hefei 230026, China}
\affiliation{Shanghai Branch, CAS Center for Excellence in Quantum Information and Quantum Physics, University of Science and Technology of China, Shanghai 201315, China}
\affiliation{Shanghai Research Center for Quantum Sciences, Shanghai 201315, China}
	
\date{\today}


\begin{abstract}
Quantum machine learning is expected to be one of the first practical applications of near-term quantum devices. Pioneer theoretical works suggest that quantum generative adversarial networks (GANs) may exhibit a potential exponential advantage over classical GANs, thus attracting widespread attention. However, it remains elusive whether quantum GANs implemented on near-term quantum devices can actually solve real-world learning tasks. Here, we devise a flexible quantum GAN scheme to narrow this knowledge gap.~\rrr{In principle, this scheme has the ability to complete image generation with  high-dimensional features and could harness quantum superposition to train multiple examples in parallel.} For the first time, we experimentally achieve the learning and generating of real-world hand-written digit images on a superconducting quantum processor. Moreover, we utilize a gray-scale bar dataset to exhibit competitive performance between quantum GANs and the classical GANs based on multilayer perceptron and convolutional neural network architectures, respectively, benchmarked by the Fréchet Distance score. Our work provides guidance for developing advanced quantum generative models on near-term quantum devices and opens up an avenue for exploring quantum advantages in various GAN-related learning tasks.
\end{abstract}

\maketitle

	\begin{figure*}[tbp]
\centering\includegraphics[width=.85\textwidth]{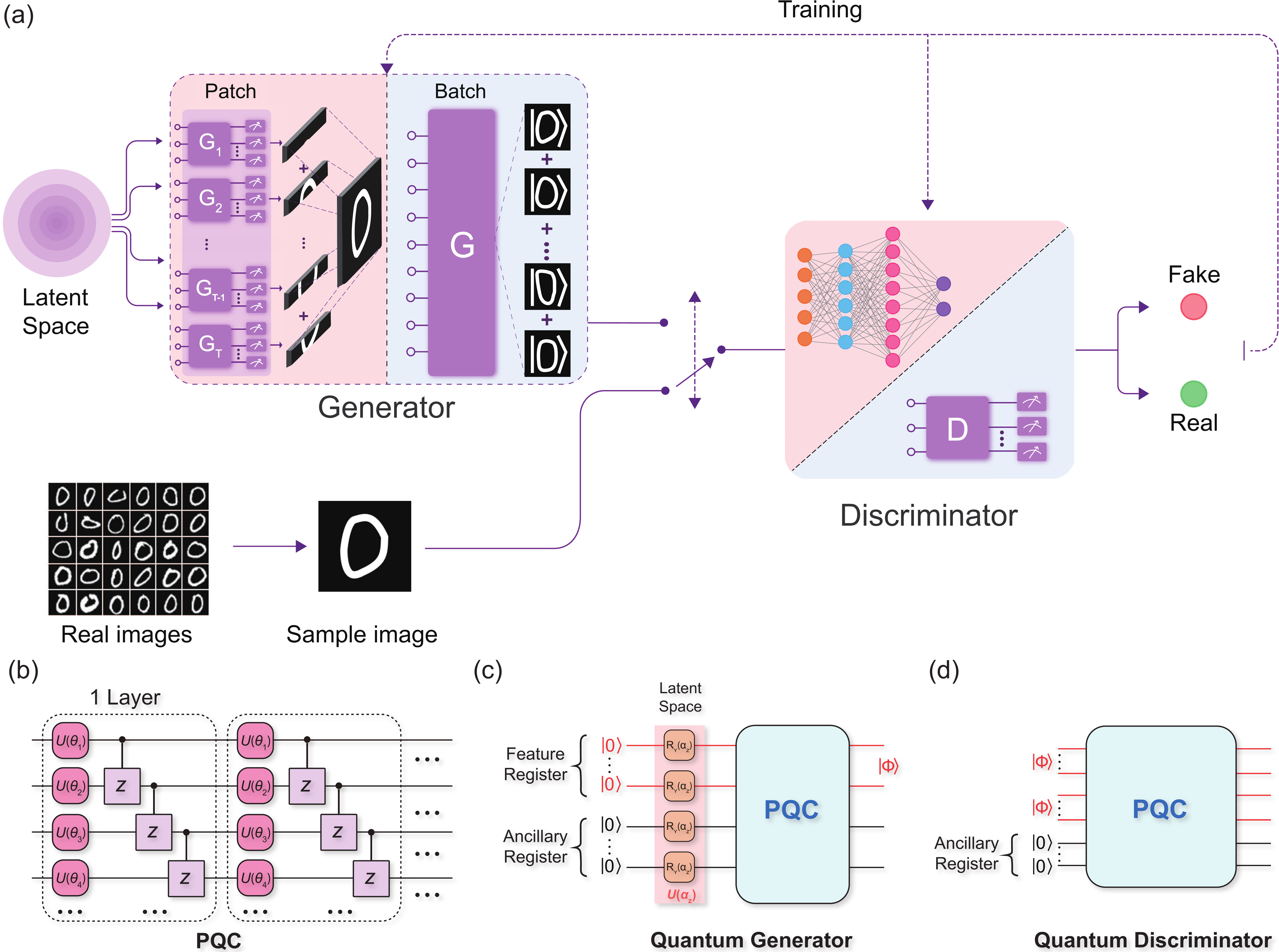}

\caption{\small{\textbf{The resource-efficient  quantum GAN scheme.} (a) The proposed quantum GANs scheme contains a quantum generator $G$ and a discriminator $D$, which can either be classical or quantum. The mechanism of quantum patch GAN  is as follows. First, the latent state $\ket{\bm{z}}$ sampled from the latent space is input into quantum generator $G$  formed by $T$ sub-generators (highlighted in pink region), where each  $G_t$ is built by a PQC $U_{G_t}(\bm{\theta}_t)$. Next, the generated image is acquired  by measuring the generated states $\{U_{G_t}(\bm{\theta}_t)\ket{\bm{z}}\}_{t=1}^T$  along the computation basis. Subsequently, the patched generated image and the real image are input into the classical discriminator $D$ (highlighted in pink region) in sequence.  Finally, a classical optimizer uses the classified results as the output of $D$ to update trainable parameters for $G$ and $D$. This completes one iteration. The mechanism of quantum batch GAN is almost identical to the quantum patch GAN, except for three modifications: 1) we set $T=1$ and introduce the quantum index register into $G$ (highlighted in blue  region); 2) the generated state $U_{G}(\bm{\theta})\ket{\bm{z}}$ directly operates with quantum discriminator $D$ implemented by PQC (highlighted in blue region), where the output is acquired by a simple measurement; and 3) the real image is encoded into the quantum state to operate with $D$. (b) The implementation of PQC used in the quantum generator and quantum discriminator. (c) The machinery of quantum generators employed in quantum patch and batch GANs. For quantum batch GAN, an index register with extra operations should be involved when the batch size is larger than one. (d) The quantum discriminator employed in the quantum  batch GAN. To attain nonlinear property, two generated states are fed into the quantum discriminator simultaneously.
}}
\label{fig:GTN}
\end{figure*}

State-of-the-art quantum computing systems are now  stepping into the era of  \textit{Noisy Intermediate-Scale Quantum}  (NISQ) technology \cite{preskill2018quantum,arute2019quantum,huang2020superconducting,zhong2020quantum}, which promises to address challenges in quantum computing and to deliver useful applications in specific scientific domains in the near term. The overlap  between quantum information and machine learning has emerged as   one of the most encouraging  applications for quantum computing, namely, quantum machine learning \cite{biamonte2017quantum}. Both theoretical and experimental evidences  suggested  that quantum computing may significantly  improve machine learning performance well beyond that    achievable with their classical counterparts  \cite{biamonte2017quantum, lloyd2014quantum, lloyd2016quantum, rebentrost2014quantum, dunjko2018machine, cai2015entanglement, huang2018demonstration, havlivcek2019supervised,liu2019hybrid,cong2019quantum}.

Generative adversarial networks (GANs) are at the forefront of the generative learning and  have been widely used for image processing, video processing, and molecule development~\cite{goodfellow2014generative}. Although GANs have achieved wide success, the huge computational overhead makes them approach the limits of Moore’s law. For example,  Big GAN with $158$ million parameters is trained to generate $512\times 512$ pixel images  using $14$ million  examples and $512$ TPU for two days~\cite{brock2018large}. Recently, theoretical works show that quantum generative models may exhibit an exponential advantage over classical counterparts~\cite{lloyd2018quantum,gao2018quantum,romero2019variational}, arousing widespread research interest in theories and experiments of quantum GANs~\cite{lloyd2018quantum,dallaire2018quantum, hu2019quantum,zoufal2019quantum, kiani2021quantum, chakrabarti2019quantum}. Previous experiments of quantum GANs on digital quantum computers, hurdled by algorithm development and accessible quantum resources, mainly focus on the single-qubit quantum state generation and quantum state loading~\cite{hu2019quantum, zoufal2019quantum}, e.g., finding a quantum channel to approximate a given single-qubit quantum state~\cite{hu2019quantum}. Such a task can be regarded as the approximation of a low-dimensional distribution with an explicit formulation. However, the explicit formula implies that these studies cannot be treated as general generative tasks, since the data space structure is exactly known. A crucial question that remains to be addressed in quantum GAN is whether current quantum devices have the capacity for real-world generative learning, which is directly related to it's practical application on near-term quantum devices.

Here we develop a resource-efficient quantum GAN scheme to answer the above question. \rrr{Our proposal principally supports to use limited quantum resources to accomplish large-scale generative learning tasks. Besides, the proposed scheme has the potential to train multiple examples in parallel given sufficient quantum resources.} We experimentally implement the scheme on a superconducting quantum processor to accomplish the generative task of real-world hand-written digit image \cite{Dua:2019}, a commonly used data in the machine learning community. Moreover, we show that quantum GAN has the potential advantage of reducing training parameters, and can achieve comparable performance with some typical classical GANs.

Following the routine of GANs \cite{goodfellow2014generative,lloyd2018quantum}, our proposal exploits a two-player minimax game between  a generator $G$ and a discriminator $D$. Given a latent vector $\bm{z}$ sampled from a certain distribution,   $G$  aims to output the generated data $G(\bm{z})\sim \Pro_g(G(\bm{z}))$ with $\Pro_g(G(\bm{z}))\approx \Pro_{data}(\bm{x})$ to fool $D$. Meanwhile,  $D$ tries to distinguish  the true example $\bm{x}  \sim \Pro_{data}(\bm{x})$ from    $G(\bm{z})$. Unlike classical GANs, the generator or discriminator in quantum GANs is constructed by quantum circuits. More precisely, denote that the deployed quantum device has $N$-qubits with $O(poly(N))$ circuit depth, and the feature dimension of the training example is $M$. We devise two flexible strategies, i.e., the patch strategy and batch strategy, that enable our quantum GANs to adequately exploit the supplied resources under the setting $N<\lceil\log M \rceil$ and $N> \lceil\log M \rceil$, respectively. The design of quantum patch GAN aims to use insufficient quantum resources to generative high-dimensional features, while the quantum batch GAN can be used for parallel training given sufficient resources. In this way, our proposal could flexibly adapt and maximally utilize accessible quantum resources.

\begin{figure*}[tbp]
\centering\includegraphics[width=0.85\textwidth]{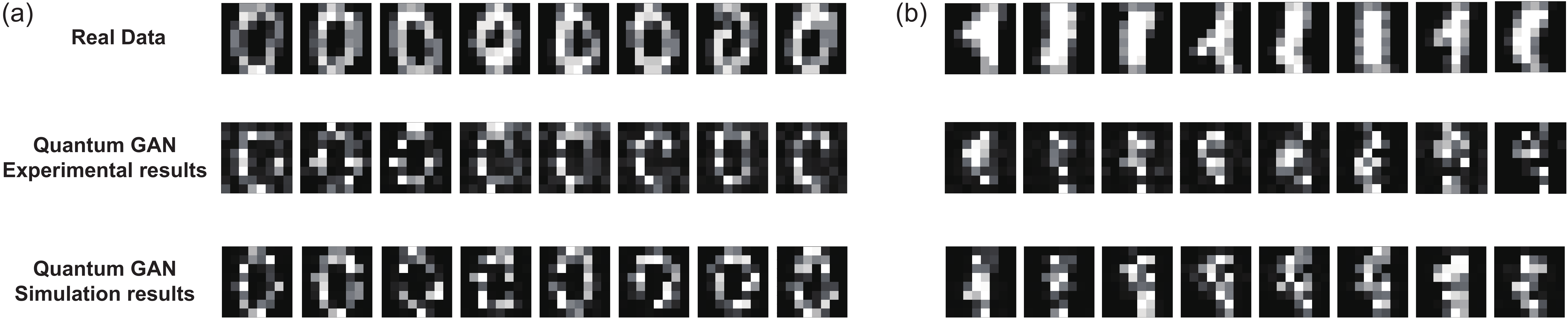}
		\caption{\small{\textbf{Hand-written digit image generation.}  (a) and (b) show the experimental results for the handwritten digit `0' and  `1', respectively.  From top to bottom, the first row illustrates real data examples, the second and third rows show the examples generated by quantum patch GAN trained using a superconducting processor and noiseless numerical simulator, respectively. The number of parameters for quantum generator is set to $100$, and the total number of iterations is about $350$.}}
			\label{fig:digit}
\end{figure*}

The quantum  patch GAN with $N < \lceil\log M \rceil$  consists of the quantum generator and the classical discriminator. A potential benefit of the quantum generator is that it may possess stronger expressive power to fit data distributions compared with classical generators. This is  supported  by  complexity theory with $P\subseteq BQP$ \cite{bernstein1997quantum},  and theoretical evidences showing that  certain  distributions generated by   quantum circuits can not be efficiently simulated by classical circuits unless the polynomial hierarchy collapses \cite{bremner2010classical,aaronson2011computational,bravyi2018quantum}.
 Fig.~\ref{fig:GTN}   illustrates the   implementation of quantum patch  GAN, where  the  patch strategy is applied to manipulate large $M$ with small $N$.  Specifically, the quantum generator $G$  is composed of a set of sub-generators $\{G_t\}_{t=1}^T$,  where each $G_t$ refers to a parameterized  quantum circuit (PQC) $U_{G_t}(\bm{\theta}_t)$. The aim of $G_t$  is  to output  a state $\ket{G_t(\bm{z})}$ with  $\ket{G_t(\bm{z})}=U_{G_t}(\bm{\theta}_t)\ket{\bm{z}}$ that represents a specific portion of the high-dimensional feature vectors. All sub-generators that scale with $T\sim O(\lceil\log M\rceil/N)$ can either be effectively built on distributed quantum devices to  train in parallel or on a single quantum device to train in sequence.   The generated example $\tilde{\bm{x}}$ is obtained by measuring $T$ states $\{\ket{G_t(\bm{z})}\}_{t=1}^T$   along the computation basis.  Given $\tilde{\bm{x}}$ and $\bm{x}$,  the loss function $\mathcal{L}$ employed to optimize the   trainable parameters $\bm{\theta}$ and  $\bm{\gamma}$ for $G$  and  $D$ yields
 \begin{equation}\label{eqn:loss}
 \begin{aligned}
 &\min_{\bm{\theta}}\max_{\bm{\gamma}} \mathcal{L}(D_{\bm{\gamma}}(G_{\bm{\theta}}(\bm{z})),D_{\bm{\gamma}}(\bm{x})),
\end{aligned}
\end{equation}
where $ \mathcal{L}(D_{\bm{\gamma}}(G_{\bm{\theta}}(\bm{z})),D_{\bm{\gamma}}(\bm{x}))= \mathbb{E}_{\bm{x}\sim \Pro_{data}(\bm{x})} [\log D_{\bm{\gamma}}( \bm{x})] +\mathbb{E}_{\bm{z}\sim \Pro(\bm{z})} [\log(1-D_{\bm{\gamma}}(G_{\bm{\theta}}(\bm{z}))]$, $\Pro_{data}(\bm{x})$ refers to the distribution of training dataset, and $\Pro(\bm{z})$ is the probability distribution of the latent variable $\bm{z}$. The concept of patching enables our quantum GAN to complete image generation task with a large $M$ using limited quantum resources. Although the usage of multiple sub-generators differs from the classical case, we can easily prove that quantum patch GAN can converge to Nash equilibrium in the optimal case (see Supplemental Material).

To evaluate performance of the quantum patch GAN, we implement it on a superconducting quantum processor to accomplish the real-world hand-written digit image generation for `0' and `1'. Specifically, the superconducting quantum processor has 12 transmon qubits on a 1D chain, and up to 6 adjacent qubits are chosen in the entire experiment. The average fidelities of single-qubit gates and controlled-$Z$ gate are approximately $0.9994$ and $0.985$, respectively. In addition, two training datasets are collected from the optical recognition of handwritten digit dataset \cite{Dua:2019}. Each training example is an $8\times 8$ pixel image with $M=64$. In the experimental settings for quantum patch GAN, we set $T=4$, $N=5$, and the total number of trainable parameters is 100. As shown in Fig.~\ref{fig:digit}, the experimental quantum GAN output similar quality images to the simulated quantum GAN, suggesting that our proposal is insensitive to noise at our current noise levels and for this system size.

\begin{figure*}[tbp]

\includegraphics[width=0.85\textwidth]{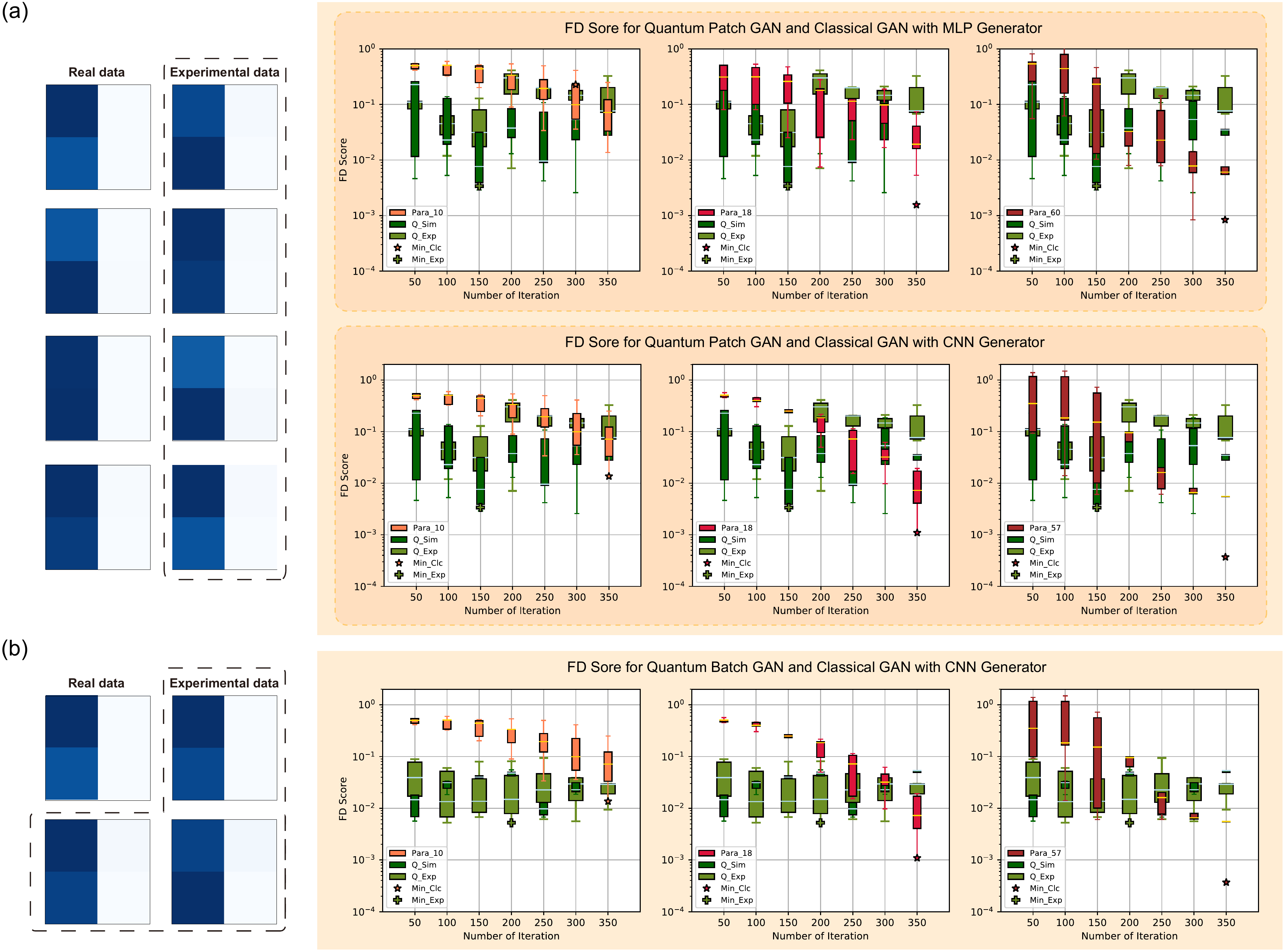}			\caption{\small{\textbf{\textbf{Gray-scale bar  image generation.}}
(a) Experiment results of quantum patch  GAN for $2 \times 2$ image dataset. The left panel illustrates real examples and generated examples. The right panel, highlighted in the yellow region, shows box plots that illustrate the FD scores achieved by different generative models. A lower FD score equates to better performance of the generative model. Specifically, we set the number of trainable parameters for $G$ as $N_p=9$,   the number of iterations as $350$,  sample $1000$ generated examples to evaluate the FD score after every $50$ iterations, and repeat each setting five times to collect statistical results. The label `para\_$\ell$' refers to the FD score of the classical GAN employing the generator with $\ell$ trainable parameter. The labels `Q\_exp' and `Q\_sim' refers to FD scores  of quantum GANs built using a quantum processor and noiseless numerical simulator, respectively. The labels `Min\_Clc' and `Min\_Exp' represents the achieved best FD scores for classical and quantum GANs, respectively. The left FD score plot compares the performance between classical and quantum GAN where  $\ell$ approximates to $N_p$, i.e., $\ell=10$, and the results show that quantum GAN have a better performance than GAN-MLP and GAN-CNN with similar number of parameters. The middle and right FD score plots show the required value $\ell$, i.e., $\ell=18 (18)$ and $\ell = 60 (57)$ for GAN-MLP (GAN-CNN), which enables the classical GANs to achieve the comparable and even better performance over quantum GANs. The performance is evaluated by the average score (middle line of the shaded  box) and the minimal FD score.
(b) Experiment results of quantum  batch GAN for $2 \times 2$ image dataset.
The three plots indicate that the quantum batch GAN could achieve a similar performance to the quantum patch GAN.
}}
	\label{fig:bar_FID}
\end{figure*}

Recall that the aim of GANs, as a kind of generative models, is to explore the probability distribution of observed samples. To accurately evaluate the well-trained generative models, we  intend to use quantitative metrics to measure the distance between real and generated distributions. However, the hand-written digit dataset is not a good choice to achieve this goal, hampered by its limited size and the implicit  distribution. With this regard, we construct a synthetic dataset, as so-called the gray-scale bar image dataset. Note that all images in this dataset are composed of simple pattern and sampled from an explicit distribution. Fig.~\ref{fig:bar_FID} exhibits some examples of gray-scale bar images. Utilizing the specific distribution, we can easily acquire an unlimited number of data samples for both training and test. Next, we use the Fréchet Distance (FD) score \cite{dowson1982frechet, heusel2017gans} to directly measure the Fréchet distance (i.e., 2-Wasserstein distance) between real and generated distributions. Such quantitative metrics could help us to comprehensively evaluate different GANs.

In the experiment, we collect a training dataset with $N_e=1000$ examples for the $2 \times 2$ gray-scale bar image dataset. Experimental parameter settings for quantum patch GAN are  $T=1$, $N=3$, and the number of trainable parameters for the quantum generator is $N_p = 9$. To benchmark the performance of the quantum patch GAN, two typical classical GANs, i.e., the classical GAN model with multilayer perceptron neural network architecture  (GAN-MLP) and the classical GAN model with convolutional neural network architecture (GAN-CNN), are employed as references. Particularly, we vary the number of generator's parameters  in these two classical GANs, and compare their performance with the quantum patch GAN. The number of parameters of the  classical discriminator used in the  GAN-MLP, GAN-CNN and quantum patch GAN is set as $96$. Figure~\ref{fig:bar_FID} (a) shows our experimental results. The employed two classical GANs request more training parameters than the quantum patch GAN to achieve similar FD scores. This result implies that quantum GAN has the potential advantage of reducing training parameters. Note that in our experiments, grid-search is applied to find the optimal hyper-parameters (e.g., learning rate) for classical GANs, while we did not search for these hyper-parameters for quantum GAN. 

The quantum batch  GAN with $N> \lceil \log M \rceil$  consists of both a quantum generator and discriminator (see Fig.~\ref{fig:GTN}).  As with the  quantum patch GAN, the quantum generator and discriminator play a minimax game accompanied by  the loss function $\mathcal{L}$ in Eqn.~(\ref{eqn:loss}).  The major difference to the first proposal is the way in which quantum resources are optimally utilized under the setting $N > \lceil \log (M) \rceil$. Specifically, we separate $N$ qubits into the feature register $\mathcal{R}_F$ and the index register $\mathcal{R}_I$, i.e., $\mathcal{R}_F$ with $N_F$ qubits encodes the feature information, while  $\mathcal{R}_I$ with  $N_I$ qubits records a batch of  generated/real examples. The training examples with batch size $N_e$ are encoded as $\frac{1}{{\sqrt {{N_{\rm{e}}}} }}\sum_{i} \ket{i}_I\ket{\bm{x}_i}_F$ by using amplitude encoding method. The attached index register `$I$' enables us to simultaneously manipulate $N_e$ examples to effectively acquire the  gradient  information, which dominates the computational  cost to train GAN.   Recall that classical GAN uses the mini-batch gradient descent \cite{li2014efficient} to update trainable parameters, i.e., at the $k$-th iteration, the updating rule is
	$\bm{\gamma}_k = \bm{\gamma}_{k-1} - \eta_D\sum_{i\in B_k} \nabla_{\bm{\gamma}} \mathcal{L}(G_{\bm{\theta}}(\bm{x}_i), D_{\bm{\gamma}}(G(\bm{z}_i)))$, where $B_k\subset [N_E]$ collects the indexes of a mini-batch examples.  Empirical  studies have shown  that increasing  the batch size $|B_k|$ contributes to improve performance of classical GAN, albeit at the expense of computational cost \cite{salimans2018improving,brock2018large}. In contrast to classical GAN, we show that the optimization term $\sum_{i\in B_k} \nabla \mathcal{L}(G_{\bm{\theta}}(\bm{x}_i), D_{\bm{\gamma}}(G(\bm{z}_i))) $ can be efficiently calculated in quantum GAN since we can naturally train $N_e$ examples simultaneously by using the quantum superposition (see Supplemental Material for details). This result implies a potential advantage of quantum batch GAN for efficiently processing big data. Moreover, since quantum batch GAN employs the quantum discriminator for binary classification, theoretically, measuring one qubit is enough to distinguish between `real' and `fake' images. Thus, the number of measurements required for quantum batch GAN  is quite small.

We also use the quantum batch GAN to accomplish the gray-scale bar image generation task to validate its generative capability. The experimental parameter settings are $T=1$, $N=3$, $|B_k| =1$ (or $N_I=0$), and total number of trainable parameters for the quantum generator is $N_p = 9$. We employ the quantum discriminator model proposed in Ref.~\cite{schuld2019quantum} as our quantum discriminator (see Fig.~\ref{fig:GTN} (d)). The total number of trainable parameters for the quantum discriminator is $12$. Fig.~\ref{fig:bar_FID} (b) shows that quantum batch GAN can achieve similar FD scores to the quantum patch GAN, thereby empirically showing that quantum batch GAN can be used to tackle image generation problems. We remark that the slightly degraded performance of the quantum batch GAN compared with the quantum patch GAN is mainly caused by the limited number of training parameters used in its quantum discriminator, i.e., $12$ versus $96$ in these two settings.

In conclusion, our experimental results provide the following key insights. First, we narrow the gap between quantum and classical generative learning. To our best knowledge, this is the first experimental study to generate real-world digit images on  a real quantum machine. Second, our results provide a positive signal to utilize quantum GANs to attain potential merits such as reducing the number of training parameters and improving the computation efficiency in the NISQ setting. Last, the comparison between numerical and experimental results indicates that quantum GAN is \rrr{resilient to a certain level of noise sources contained in the deployed quantum device. Noise resilience is of great importance for the realization of variational quantum algorithms on NISQ chips~\cite{sharma2020noise,mcclean2016theory}.}~

\rrr{When applying our proposal to deal with large-scale problems, some efforts may be made to  sustain its trainability and avoid barren plateaus \cite{mcclean2018barren}. It remains unknown whether the optimization of a minimax loss in Eqn.~(\ref{eqn:loss}) encounters barren plateaus. Namely, how the varied loss functions and optimization methods affect the trainability of quantum GANs. A deep understanding of this topic enables us to devise more powerful and efficient quantum GANs. Celebrated by the versatility of our proposal, a probable approach to avoid barren plateaus is designing barren-plateaus-immune ansatz  \cite{cerezo2021cost,du2020quantum,pesah2020absence,zhang2020toward} instead of the hardware-efficient ansatz to implement the quantum generator or discriminator. In addition, the adaptivity of the proposed quantum patch GAN enlightens a novel way to conquer barren plateaus and noise. Through tailoring the large-size problems into multiple small-size problems, the trainability of quantum patch GAN may be warranted. In light of these discussions, an intrigued  research direction is experimentally exploring the trainability of quantum GANs on large-scale datasets.}

\rrr{We would like to point out that although quantum GANs can partially adapt to the imperfection of quantum systems, a fundamental principle to enhance the performance of quantum GANs is continuously promoting the quality of quantum processors, e.g., a larger number of qubits, a more diverse connectivity, lower system noise, and longer decoherence time. For this purpose, we will delve to implement quantum GANs on more advanced quantum computers to accomplish complex real-world generation tasks to seek their potential advantages.}

\begin{acknowledgments}
The authors thank the Laboratory of Microfabrication, University of Science and Technology of China, Institute of Physics CAS, and National Center for Nanoscience and Technology for supporting the sample fabrication. The authors also thank QuantumCTek Co., Ltd., for supporting the fabrication and the maintenance of room-temperature electronics. We thank Johannes Majer for helpful discussion. \textbf{Funding}: This research was supported by the National Key Research and Development Program of China (Grants No. 2017YFA0304300), NSFC (Grants No. 11574380, No. 11905217), the Chinese Academy of Science and its Strategic Priority Research Program (Grants No. XDB28000000), the Science and Technology Committee of Shanghai Municipality, Shanghai Municipal Science and Technology Major Project (Grant No.2019SHZDZX01), and Anhui Initiative in Quantum Information Technologies. H.-L. H. acknowledges support from the Youth Talent Lifting Project (Grant No. 2020-JCJQ-QT-030), National Natural Science Foundation of China (Grants No. 11905294), China Postdoctoral Science Foundation, and the Open Research Fund from State Key Laboratory of High Performance Computing of China (Grant No. 201901-01).
\end{acknowledgments}


\bibliographystyle{apsrev4-1}
\bibliography{myref2}

\newpage 	
\renewcommand{\thefigure}{M\arabic{figure}}	
\setcounter{figure}{0}

\end{document}


	\title{Supplemental Material for \\ ``Experimental Quantum Generative Adversarial Networks for Image Generation''}
	
	

	\maketitle


\section{SM (A): Preliminaries }
Here we briefly introduce the essential backgrounds used in this paper  to facilitate both physics and computer science communities.  Please see \cite{nielsen2010quantum,goodfellow2016deep} for more elaborate descriptions. In particular, we define necessary notations  and  exemplify a typical deep neural network, i.e., fully-connected neural network, in the first two subsections. We then present a classical GAN and illustrate its working mechanism.
Afterwards, we provide the definition of box-plot, which is employed to analyze the performance of the generated data. Ultimately, we recap the parameter quantum circuits, as the building block of quantum GAN.
\subsection{Notations}
We unify some basic notations used throughout the whole paper.  We denote the set  $\{1,2,..., n\}$ as $[n]$.  Given a vector ${\mathbf{v}}\in \mathbb{R}^{n}$,  ${\mathbf{v}}_i$ or ${\mathbf{v}}(i)$  represents the $i$-th entry of $\mathbf{v}$ with $i\in [n]$ and $\|{\mathbf{v}}\|$ refers to the $\ell_2$ norm of $\mathbf{v}$ with $\|{\mathbf{v}}\| =\sqrt{ \sum_{i=1}^n {\mathbf{v}}_i^2}$.  The notation $\mathbf{e}_i$ always refers to the $i$-th unit basis vector.   We use Dirac notation that is broadly used in quantum computation to write the computational basis $\bm{e}_i$ and $\bm{e}_i^{\top}$ as $\ket{i}$ and $\bra{i}$. A pure quantum state $\ket{\psi}$ is represented by a unit vector, i.e., $\braket{\psi|\psi}=1$. A mixed state of a quantum system $\rho$ is denoted as $\rho = \sum_{i}p_i \ket{\phi_i}\bra{\phi_i}$ with $\sum_i p_i=1$ and $\Tr(\rho) = 1$.  The symbol `$\circ$' is used to represent the composition of functions, i.e., $f\circ g (x) = f(g(x))$.  The observable $\bm{x}$ sampled from the certain distribution $p(\bm{x})$ is denoted as $\bm{x}\sim \Pro(\bm{x})$. Given two sets $A$ and $B$,   $A$ minus $B$ is written as $A\setminus B$. We employ the floor function  that takes real number $x$ and  outputs the greatest integer $x':=\lfloor x \rfloor$ with $x'\leq x$. Likewise, we employ the ceiling  function  that takes real number $x$ and  outputs the least integer $x':=\lceil x \rceil$ with $x'\geq x$.

\subsection{Fully-connected neural network}\label{sec:FCNN}\label{subsec:FCNN and Its Optimization}
Fully-connected neural network (FCNN), as the biologically inspired computational model, is the workhorses of deep learning \cite{goodfellow2016deep}.   Various advanced deep learning models are  devised by combing FCNN with additional techniques, e.g.,  convolutional layer \cite{krizhevsky2012imagenet}, residue  connections \cite{he2016deep}, and  attention mechanisms \cite{vaswani2017attention}.  FCNN and its variations have achieved state-of-the-art   performance over other computation models in many machine learning tasks.

The basic architecture of FCNN is shown in the left panel of Fig.~\ref{fig:FCNN}, which includes an  input layer,  $L$ hidden layers with $L\geq 1$, and an output layer. The node in each layer is called `neuron'.  A typical feature of FCNN is that a neuron at $l$-th layer is only allowed to connect to a neuron at $(l+1)$-th layer.  Denote that the number of neurons and the output of $l$-th layer as $n_l$ and $\bm{x}^{(l)}$, respectively.  Mathematically, the output of $l$-th layer can be treated as a vector $\bm{x}^{(l)}\in\mathbb{R}^{n_{l}}$ and each neuron  represents an entry of $\bm{x}^{(l)}$.  Let the connected edge between the $l$-th layer and $(l+1)$-layer be $\bm{\Theta}^{(i)}$.  The connected edge refers to a weight matrix $\bm{\Theta}^{(l)}\in\mathbb{R}^{n_l\times n_{l+1}}$. The calculation rule for the $j$-th  neuron at $l+1$-th layer $\bm{x}^{(l+1)}(j)$ is demonstrated in the right panel of Fig. \ref{fig:FCNN}. In particular, we have $
	\bm{x}^{(l+1)}(j) := g_l(\bm{x}^{(l)})=f(\bm{\Theta}^{(l)}(j,:)\bm{x}^{(l)})$, where $f(\cdot)$ refers to the activation function. Example activation functions include the sigmoid function with 	$f(\bm{x}) =  (1+e^{\bm{x}})$ and the Rectified Linear Unit (ReLU) function  with $f(\bm{x}) = \max(\bm{x},0)$ \cite{goodfellow2016deep}.  Since the output of the $l$-th layer is used as an input for the $l+1$-th layer,  an $L$-layers FCNN model is given by
	\begin{equation}
		\bm{x}^{(out)} =g_L\circ ...\circ g_i \circ...  \circ g_1(\bm{x}^{(in)})~,
	\end{equation}
where $\bm{x}^{(in)}$ and $\bm{x}^{(out)}$ refer to the input and output vector, and $g_l$ with $l\in[L]$ is parameterized by $\{\bm{\Theta}^{(l)}\}_{l=1}^L$.  In the training process, the weight matrices  $\{\bm{\Theta}^{(l)}\}_{l=1}^L$ are optimized to minimize a predefined  loss function $\mathcal{L}_{\bm{\Theta}}(\bm{x}^{(out)}, \bm{y} )$ that measures the difference between the output $\bm{x}^{(out)}$ and the expected result $\bm{y}$.

 \begin{figure}[!ht]
	\centering
	\includegraphics[width=0.8\textwidth]{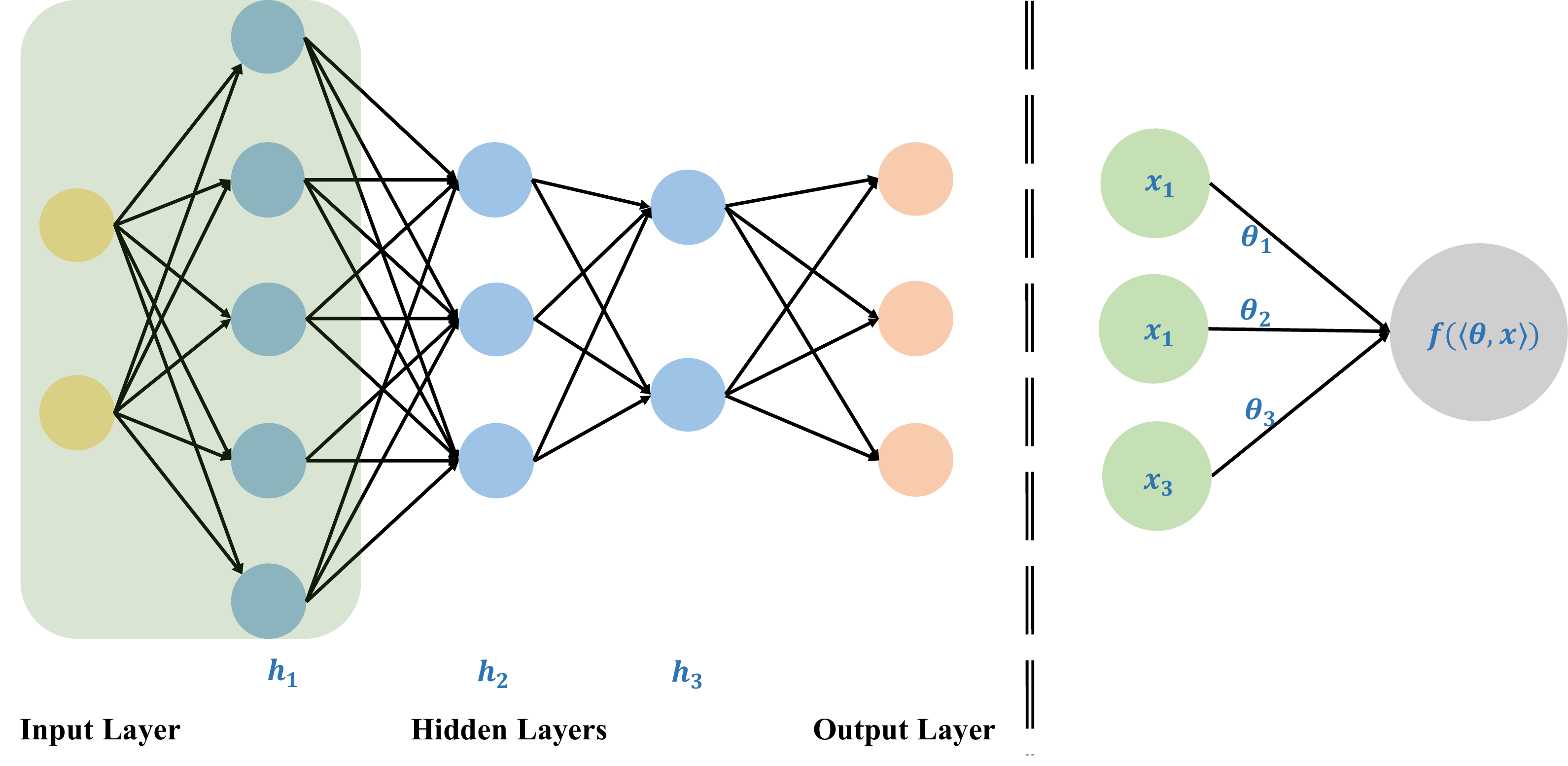}
	\caption{\small{\textbf{An example of FCNN}. The left panel illustrates the basic structure of FCNN that consists of an input layer, one hidden layer, and an output layer. In the green  region, the number of neurons for the input layer and the first hidden layer is $2$ and $5$, respectively.  The right panel shows the calculation rule of a single neuron. The neuron, highlighted by gray region, is calculated by  $f(\langle \bm{\theta},\bm{x}\rangle )$, where   $\bm{\theta}$ represents the weight, $\bm{x}$ refers to the outputs of the green neurons, and $f(\cdot)$ is the predefined activation function.     }}
	\label{fig:FCNN}
\end{figure}

In deep learning, the most effective  method to optimize trainable weight matrix $\bm{\Theta}$ with $\bm{\Theta}=[\bm{\Theta}^{(1)},...,\bm{\Theta}^{(L)}]$ is gradient descent \cite{bishop2006pattern}. From the perspective of how many training examples are  used to compute the gradient, we can mainly divide various gradient descent methods into three categories, i.e.,  stochastic gradient descent, batch gradient descent, and mini-batch gradient descent \cite{ruder2016overview}. For the sake of simplicity,   we explain  the mechanism of these three methods in the binary classification task.  Suppose that the given dataset $\mathcal{D}$ consists of $M$ training examples with $\mathcal{D}=\{(\bm{x}_i, \bm{y}_i)\}_{i=1}^{M}$ and $\bm{y}_i \in \{0,1\}$. Let $\mathcal{L}$ be the loss function to be optimized and $\eta$ be the learning rate.  The batch gradient descent computes the gradient of the loss  function of  the whole dataset at each iteration, i.e., the optimization at $k$-th iteration step is
\begin{equation}
	  \bm{\Theta}_{k} = \bm{\Theta}_{k-1}-\eta\frac{1}{M} \sum_{i=1}^M  \nabla_{\bm{\Theta}} \mathcal{L}(\bm{x}_i, \bm{y}_i) ~.
\end{equation}
The stochastic gradient descent (SGD), in contrast to batch gradient descent,  performs a parameter update by using  single training example that is randomly sampled from the dataset $\mathcal{D}$. The mathematical representation is
\begin{equation}
\bm{\Theta}_{k} = \bm{\Theta}_{k-1}-\eta\nabla_{\bm{\Theta}}   \mathcal{L}(\bm{x}_i, \bm{y}_i) ~ \text{with}~ \bm{x}_i\in  \mathcal{D}~.	
\end{equation}
Mini-batch gradient descent employs $M'$ training examples that are randomly sampled from $\mathcal{D}$ with $M'\ll M$ to update parameters at each iteration. In particular, we have
\begin{equation}\label{eqn:mini_batch}
	\bm{\Theta}_{k} = \bm{\Theta}_{k-1}-\eta\frac{1}{M'} \sum_{i=1}^{M'}  \nabla_{\bm{\Theta}}  \mathcal{L}(\bm{x}_i, \bm{y}_i) ~ \text{with}~ \{\bm{x}_i\}_{i=1}^{M'} \subset \mathcal{D}~.
\end{equation}

Celebrated by its flexibility and performance guarantees, the  mini-batch gradient descent method is prevalently employed in deep learning compared with the rest two methods \cite{ruder2016overview}.

With the aim to achieve better convergence guarantee, advanced mini-batch gradient descent methods are highly desirable. Recall that   vanilla mini-batch gradient descent defined in Eqn.~(\ref{eqn:mini_batch}) usually encounters kinds of difficulties, e.g., how to choose a proper learning rate, and how to set learning rate schedules that adjust the learning rate during training. To remedy the weakness of vanilla mini-batch gradient descent, various improved mini-batch gradient descent optimization algorithms  have been proposed, i.e., momentum methods \cite{qian1999momentum}, Adam \cite{kingma2013auto}, Adagrad \cite{duchi2011adaptive}, to name a few. Since Adam can be employed to train quantum batch GAN, we briefly introduced its working mechanism. Specifically, Adam is a method that computes adaptive learning rates for each parameter. At $k$-th iteration, let $\bm{g}_k$ be $\bm{g}_k = \frac{1}{M'} \sum_{i=1}^{M'}  \nabla_{\bm{\Theta}}  \mathcal{L}(\bm{x}_i, \bm{y}_i)$. Define $\bm{m}_k$ and $\bm{v}_k$ as $\bm{m}_k=\beta_1\bm{m}_{k-1} + (1-\beta_1)\bm{g}_t$ and $\bm{v}_k=\beta_2\bm{v}_{k-1} + (1-\beta_2)\bm{g}_t^2$, where $\beta_1$ and $\beta_2$ are constants with default settings $\beta_1=0.9$ and $\beta_2=0.999$.  The the update rule of Adam is
\begin{equation}
	\bm{\Theta}_{k+1} = \bm{\Theta}_{k}-\frac{\eta}{\sqrt{\frac{\bm{v}_k}{1-\beta_2^k}}+\epsilon}\frac{m_k}{1-\beta_1^k} ~,
\end{equation}
where $\epsilon$ is the predefined tolerate rate with default setting $\epsilon=10^{-8}$.

\subsection{Generative adversarial network }
Generative model takes a training dataset $\mathcal{D}$ with limited examples that are sampled from distribution $\Pro_{data}$ and aims to estimate $\Pro_{data}$  \cite{goodfellow2016deep}. 
Generative adversarial network (GAN), proposed by Goodfellow in 2014 \cite{goodfellow2014generative}, is one of the most  powerful generative models.  Here we   briefly review  the theory of  GAN and  explain how to   use  FCNN   to implement GAN.

  The fundamental mechanism of GAN and its variations \cite{arjovsky2017wasserstein, mirza2014conditional,zhang2017stackgan,makhzani2015adversarial,deng2019unsupervised,mao2017least,wang2018perceptual} can be summarized as follows.  GAN sets up a  two-players  game, where the first player is called the generator $G$ and the second player is  called the discriminator $D$. The generator  $G$ creates data that pretends to come from $\Pro_{data}$ to fool the discriminator $D$, while $D$ tries to distinguish the fake generated data from the real training data.  Both $G$ and $D$ are typically implemented by deep neural networks, e.g., fully connected neural network and convolution neural network \cite{krizhevsky2012imagenet,sainath2015convolutional}.  From the mathematical perspective, $G$ and $D$ corresponds to two a differentiable  functions. The input and output of $G$ are a latent variables $\bm{z}$  and  an observed variable $\bm{x}'$, respectively, i.e.,  $G:G(\bm{z}, \bm{\theta})\rightarrow \bm{x}'$ with $\bm{\theta}$ being trainable parameters for $G$. The employed  latent variable $\bm{z}$ ensures  GAN to be a structured probabilistic model \cite{goodfellow2016deep}. In addition, the input and output of $D$  are the given example (can either be the generated data $\bm{x'}$ or the real data $\bm{x}$ ) and the binary classification result (real or fake), respectively. Mathmatically,  we have $D : D(\bm{x}, \bm{\gamma})\rightarrow (0,1) $ with $\bm{\gamma}$ being trainable parameters for $D$. If the distribution $\Pro(G(\bm{z}))$ learned by $G$ equals to  the real data distribution, i.e., $\Pro(G(\bm{z})) = \Pro(\bm{x})$, then the probability that  discriminator  predicts all inputs as real inputs is  $50\%$. This unique solution that $D$ can never discriminate between the generated data and the real data is called Nash equilibrium \cite{goodfellow2014generative}.

The training process of GANs involves both finding the parameters of a discriminator $\bm{\gamma}$ to maximize the classification accuracy, and finding the parameters of a generator $\bm{\theta}$ to maximally confuse the discriminator. The two-player game set up for GAN is evaluated by a loss function $\mathcal{L}(D_{\bm{\gamma}}(G_{\bm{\theta}}(\bm{z})),D_{\bm{\gamma}}(\bm{x}))$ that depends on both the generator and the discriminator. For example, by labeling the true data as $1$ and the fake data as $0$, the training procedure of original GAN can be treated as:
\begin{equation}\label{eqn:loss}
\min_{\bm{\theta}}\max_{\bm{\gamma}} \mathcal{L}(D_{\bm{\gamma}}(G_{\bm{\theta}}(\bm{z})),D_{\bm{\gamma}}(\bm{x})):=  \mathbb{E}_{\bm{x}\sim \Pro_{data}(\bm{x})} [\log D_{\bm{\gamma}}( \bm{x})] +\mathbb{E}_{\bm{z}\sim \Pro(\bm{z})} [\log(1-D_{\bm{\gamma}}(G_{\bm{\theta}}(\bm{z}))]~,
\end{equation}
where  $\Pro_{data}(\bm{x})$ refers to the distribution of training dataset, and $\Pro(\bm{z})$ is the probability distribution of the latent variable $\bm{z}$.
 During training, the parameters of two models are updated iteratively using  gradient descent methods \cite{boyd2004convex}, e.g., the vanilla mini-batch gradient descent and Adam introduced in Subsection \ref{sec:FCNN}.   When  parameters $\bm{\theta}$ of $G$  are updated, parameters $\bm{\gamma}$  of $D$ are keeping fixed. 	

To overcome the training hardness, e.g., the optimized parameters generally converge to the saddle points, various GANs are proposed to attain better generative performance. The improved performance is guaranteed by introducing   stronger neural network models for $G$ and $D$ \cite{zhang2018self},  powerful  loss functions \cite{arjovsky2017wasserstein} and advanced optimization methods, e.g., batch normalization and spectral normalization \cite{ioffe2015batch,miyato2018spectral}.

\subsection{Box plot}
 Box-plot, as a popular statistical tool, is made up of five components  to give a robust summary of the distribution of a dataset \cite{pande2014numeta}. As shown in Fig.~\ref{fig:box-plot}, the five components are  the median,  the upper hinge, the lower hinge, the upper extreme, and the lower extreme. Denote the first quantile as $Q_1$, the second quantile as $Q_2$, and the third quantile as $Q_3$ \footnote{$Q_1$ splits off the lowest 25\% of data from the highest 75\%.  $Q_3$ splits off the highest 25\% of data from the lowest 75\%.}. The upper (or lower) hinge represents the  $Q_3$ and $Q_1$, respectively. The median of the box-plot refers to the $Q_2$.   Let Inter quantile range ($\text{IQR}$) be $Q_3-Q_1$. The upper and lower  extreme are defined as $Q_3+1.5\text{IQR}$ and  $Q_1 - 1.5\text{IQR}$, respectively.  The data point, which is out of the region between the upper and lower  extreme, is treated as the outlier.

\begin{figure}[!ht]
\centering\includegraphics[width=0.35\textwidth ]{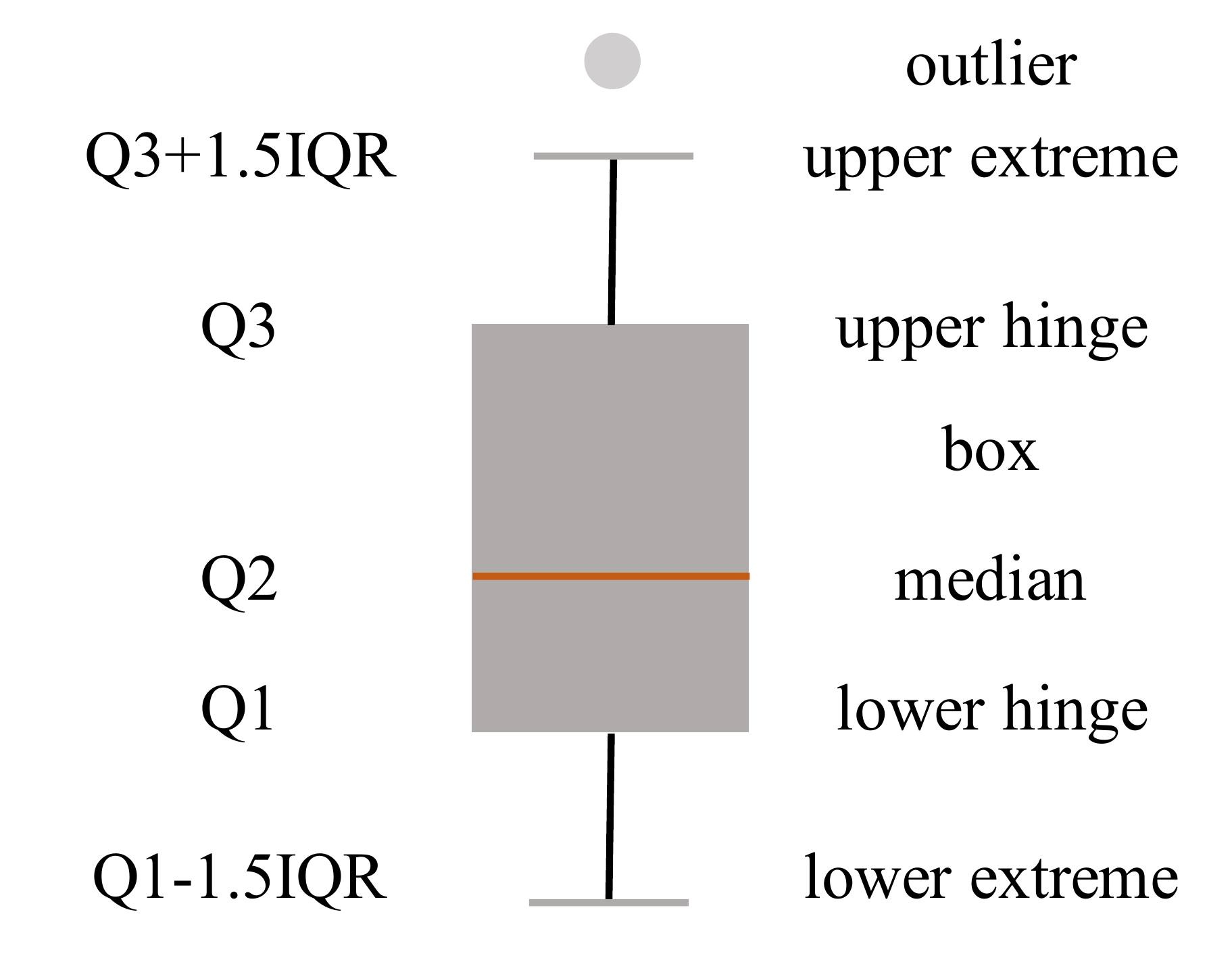}
	\caption{\small{\textbf{An example of the box-plot}. The gray circle refers to the outlier of the given  data. The orange line represents the median of the given data. The two thick gray lines correspond to the upper extreme and lower extreme, respectively. The upper edge and lower edge of the gray box stand for the upper hinge and the lower hinge, respectively.  The distance from the upper extreme to the upper hinge (or from the lower  extreme to the lower hinge) equals to $1.5\text{IQR}$. }}
	\label{fig:box-plot}
\end{figure} 	

\subsection{Parameterized quantum circuit}
Parameterized quantum circuit (PQC) is a special type of quantum circuit model that can be efficiently implemented on near-term quantum devices \cite{benedetti2019parameterized}.  The basic components of PQC are  quantum fixed two qubits gates, e.g., controlled-Z (CZ) gates, and trainable single qubit gates, e.g.,  the rotation gates $\Ry(\theta)$ along y-axis. A PQC is used to implement a unitary transformation operator $U(\bm{\theta})$ with $O(poly(N))$ parameterized quantum gates, where $N$ is the number of input qubits and  $\bm{\theta}$ is trainable parameters. The parameters $\bm{\theta}$ are updated by a classical optimizer to minimize the loss function $\mathcal{L}_{\bm{\theta}}$ that evaluates the dissimilarity between the output of PQCs and the target result.

One typical PQC is multilayer parameterized quantum circuit (MPQC), which has a wide range of applications in quantum machine learning \cite{benedetti2019generative,tacchino2019artificial,du2018implementable,hu2019quantum}. The trainable unitary operator $U(\bm{\theta})$, represented by MPQC,   is composed of $L$ layers and   each layer    has an identical arrangement of quantum gates.  Fig. \ref{fig:AD-MPQCs} (a) illustrates the general framework of MPQC.  Mathematically, we have $U(\bm{\theta}) := \prod_{l=1}^L (U_EU_{l}(\bm{\theta}))$ with $L\sim O(poly(N))$, where $U_{l}(\bm{\theta})$ is  the $l$-th  trainable layer  and $U_{E}$ is the entanglement layer. In particular,  we have  $U_{l}(\bm{\theta})=\bigotimes_{i=1}^N (U_S(\bm{\theta}^{(i,l)}))$, where $\bm{\theta}^{(i,l)}$ represents the $(i,j)$-th entry of $\bm{\theta}\in \mathcal{R}^{N\times L}$,  $U_S$ is the trainable unitary with  $U_S\in SU(2)$, e.g., the rotation single qubit gates RX, RY, and RZ.  The entangle layer $U_E$ consists of fixed two qubits gates, e.g., CNOT and CZ, where the control and target qubits can be randomly arranged.  We exemplify the implementation of  $U_l$ and $U_E$ in Fig.~\ref{fig:AD-MPQCs} (b).

\begin{figure}[!ht]
	\centering
	\includegraphics[width=0.95\textwidth]{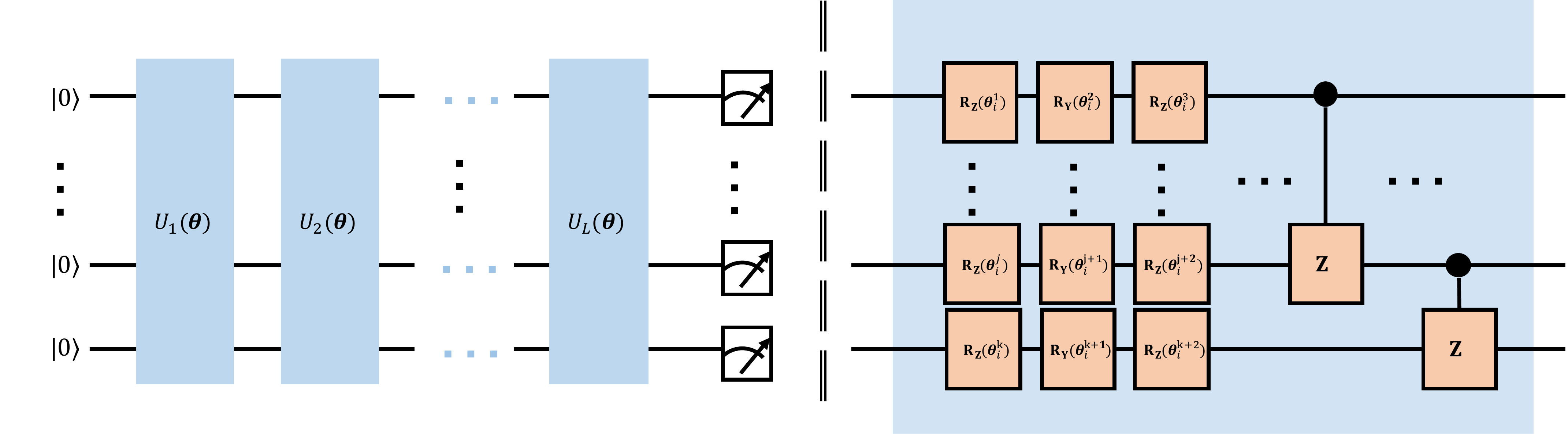}
	\caption{\small{\textbf{The implementation of MPQC}. (a) A general framework of MPQC. The trainable unitary $U_l(\bm{\theta})$ with $l\in[L]$ refers to the $l$-th layer of MPQC. The arrangement of quantum gates in each layer is identical. (b) A paradigm for the trainable unitary $U_l(\theta)$ and $U_E$. For $U_l(\theta)$, the trainable qubit gates $U_S$ are rotation single qubit gates along $Z$ and $Y$ axis. The trainable parameter refers to the rotation angle. For $U_E$, the fixed two qubits gates, i.e., CZ gates,   are applied onto the adjacent qubits.     }}
	\label{fig:AD-MPQCs}
\end{figure}

\section{SM (B): An overview of our quantum GAN}

 \begin{figure*}[htp]
			\centering   \includegraphics[width=0.65\textwidth]{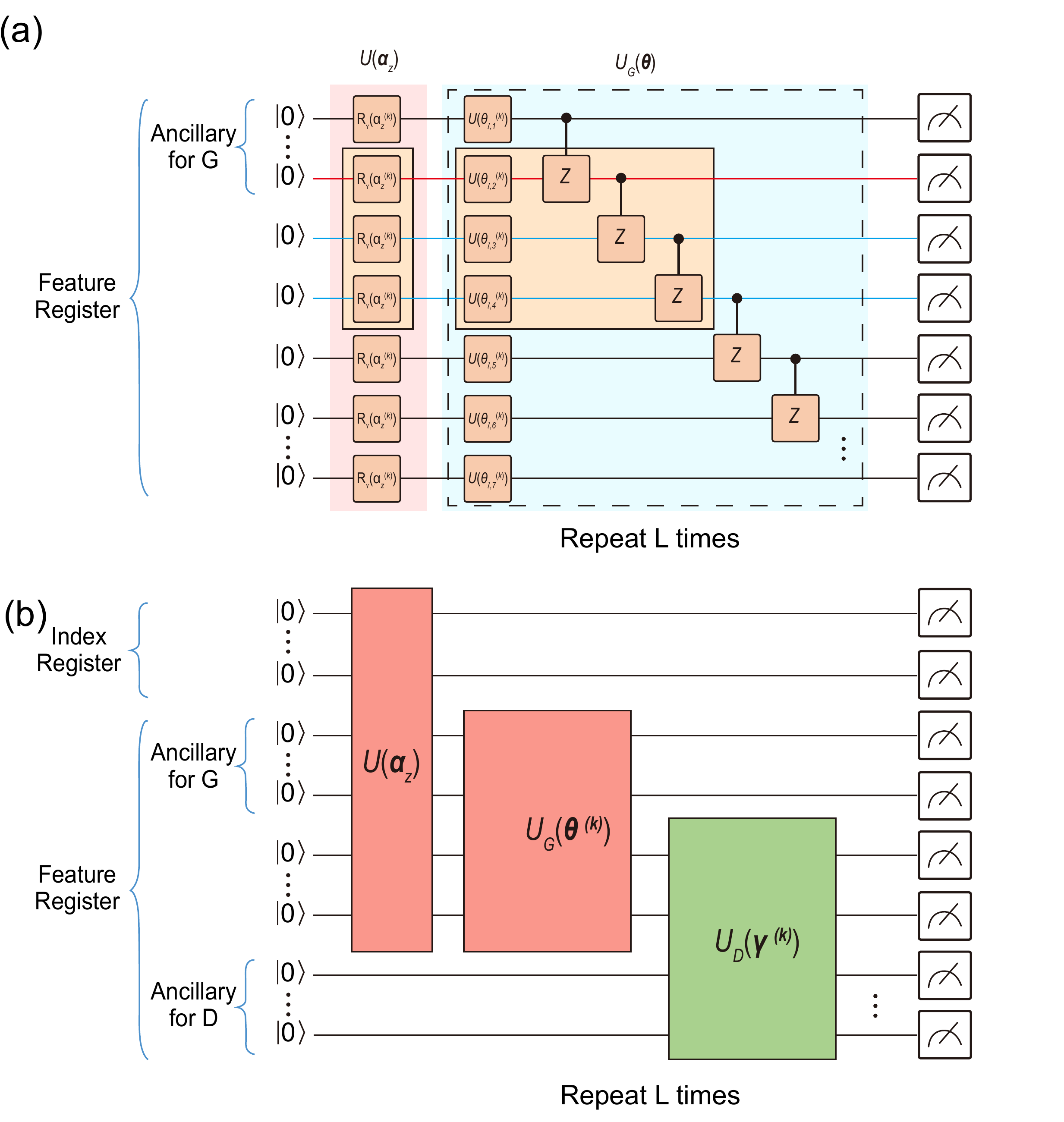}
\caption{\small{\textbf{Quantum GAN.}  (a) The implementation of $G_t$ for t  quantum patch GAN. The sub-generator $G_t$, or equivalently, the trainable unitary $U_G(\bm{\theta}_t)$, is constructed by PQC and highlighted by the  blue region with a dashed outline. Let  $U_G(\bm{\theta}_t) := \prod_{l=1}^L (U_EU_{l}(\bm{\theta}_t))$, where $U_{l}(\bm{\theta}_t):=\bigotimes_{i=1}^N (U_S(\bm{\theta}_t^{(i,l)}))$ is  the $l$-th  trainable layer,  $U_S(\bm{\theta}_t^{(i,l)})$ is the trainable unitary with  $U_S\in SU(2)$, and  $U_{E}$ is the entanglement layer with $U_{E}:=\bigotimes_{i=1}^{2i+1 \leq  N} \CZ(2i,2i+1)\bigotimes_{i=1}^{2i \leq  N} \CZ(2i-1,2i)$.  For example, we set $U_S(\bm{\theta})=\text{R}_\text{Y}(\bm{\theta})$, $L=3$, and $N=3$ to accomplish the gray-scale bar  image generation in case $m=2$, where the employed qubits are highlighted by the  blue line and the used quantum gates are highlighted by yellow region.   (b)  The main architecture of the  quantum batch GAN. A pre-trained unitary $U(\bm{\alpha}_z)$, the quantum generator $U_G(\bm{\theta}^{(k)})$, and the quantum discriminator $U_D(\bm{\gamma}^{(k)})$ are  applied to the input  state $\ket{0}^{\otimes N}$ in sequence. We adopt the same rules used in (a)  to build $U_G(\bm{\theta}^{(k)})$ and  $U_D(\bm{\gamma}^{(k)})$.} }
			\label{fig:generator}
\end{figure*}

\textbf{Quantum patch GAN.}
The three core components of quantum patch  GAN are the quantum generator, classical discriminator, and optimization rule. Here we briefly introduce the primary mechanism of these three elements, with the details presented in the SM(C).

The employed quantum generator consists of $T$ sub-generators $\{G_t\}_{t=1}^T$,  and each sub-generator is assigned to generated a specific portion of the feature vector. We exemplify the implementation of the sub-generator  $G_t$ at the  $k$-th iteration, since  the identical methods are applied to all  quantum sub-generators.    Suppose that the available quantum device has $N$ qubits, we first divide it into two parts, where the first $N_G$ qubits aim to generate a feature vector of  length $2^{N_G}$, and the remaining  $N_A$ qubits aim to conduct the nonlinear mapping, which is an essential operation in deep learning  (see the Supplementary for   details). We then prepare the  input state $\ket{\bm{z}^{(k)}}$, where the mathematical form of the input state  is  $\ket{\bm{z}^{(k)}}=(\bigotimes_{i=1}^{N} \text{R}_\text{Y} (\bm{\alpha}_z^{(k)}))\ket{0}^{\otimes N}$, where $\text{R}_\text{Y}$ refers to the rotation single qubit gate along the  y-axis and  $\bm{\alpha}_z^{(k)}$ is sampled from the  uniform distribution, e.g., $\bm{\alpha}_z^{(k)}\sim\text{unif}(0, \pi)$. Note that at each iteration, the same latent state $\ket{\bm{z}^{(k)}}$ is input into  all  sub-generators.  We then input   $\ket{\bm{z}^{(k)}}$ into  $G_t$, namely, a trainable unitary $U(\bm{\theta}_t^{(k)})$.  Figure~\ref{fig:generator}(a) shows  the implementation of  $U(\bm{\theta}_t^{(k)})$. The generated quantum state of   $G_t$ is    \begin{equation}\ket{\Psi_t^{(k)}(\bm{z})}=U(\bm{\theta}_t^{(k)})\ket{\bm{z}^{(k)}}~. \end{equation}     We finally   partially measure the generated state  $\ket{\Psi_t^{(k)}(\bm{z})}$ to obtain the classical generated result. In particular,  the $j$-th entry with $j\in[2^{N_{G_t}}]$ is  \begin{equation}\Pro_t^{(k)}(j) =\braket{\Psi_t^{(k)}(\bm{z})(\ket{j}\bra{j}\otimes \left(\ket{0}\bra{0}\right)^{\otimes N_A})\Psi_t^{(k)}(\bm{z})}~.\end{equation} Overall, the generated image $\tilde{\bm{x}}^{(k)}$ at the  $k$-th training iteration is produced by combining $T$ measured  distributions, with $\tilde{\bm{x}}^{(k)}=[\Pro_1^{(k)}, \Pro_2^{(k)},...,\Pro_T^{(k)}]\in\mathbb{R}^M$.

The employed discriminator is implemented with a  classical deep neural network, i.e., the fully-connected neural network \cite{goodfellow2016deep}. The implementation method exactly follows the classical GAN \cite{goodfellow2014generative}.   The input of  the discriminator can either be a generated image $\tilde{\bm{x}}$ or a real image $\bm{x}$ sampled from $\mathcal{D}$. The output of the discriminator $D(\bm{x})$ or $D(\tilde{\bm{x}})$ is in the range between $0$ (label `False') and $1$ (label `True').

Quantum GAN training  is analogous to classical GAN training. A loss function $\mathcal{L}$ is  employed to iteratively optimize the quantum  generator $G$ and the classical  discriminator $D$ during $K$ iterations. The mathematical form of the loss function is
\begin{equation} 	\mathcal{L}(\bm{\theta}, \bm{\gamma})= \frac{1}{M'} \sum_{i=1}^{M'}\left [ \log(D_{\bm{\gamma}}(\bm{x}^{(i)})) + \log(1-D_{\bm{\gamma}}(G_{\bm{\theta}}(\bm{z}^{(i)})) )  \right],\end{equation}
 where $\bm{x}^{(i)} \in \mathcal{D}$, $\bm{z}^{(i)}\sim \Pro(\bm{z})$, $M'$ is the size of mini-batch,  and  $\bm{\theta}$ and $\bm{\gamma}$ are trainable parameters for $G$ and $D$, respectively. The objectives of the generator and the  discriminator are  to  minimize and maximize the loss function (classification accuracy), i.e., $\max_{\bm{\gamma}}\min_{\bm{\theta}}	\mathcal{L}(\bm{\theta}, \bm{\gamma})$.   The updating rule for $G$  is  $ \bm{\theta}^{(k+1)}=\bm{\theta}^{(k)} - \eta_G*\partial_{\bm{\theta}} \mathcal{L}(\bm{\theta}^{(k)}, \bm{\gamma}^{(k)})/\partial \bm{\theta}^{(k)} $. Similarly, the updating rule for $D$ is $\bm{\gamma}^{(k+1)}=\bm{\gamma}^{(k)} +  \eta_D*\partial_{\bm{\gamma}} \mathcal{L}(\bm{\theta}^{(k)}, \bm{\gamma}^{(k)})/\partial \bm{\gamma}^{(k)}$, where $\eta_G$ ($\eta_D$) refers to the learning rate of $G$ and $D$.

\textbf{Quantum batch  GAN.} Following the same routine as the quantum patch GAN,  the  quantum  batch GAN is composed of a quantum generator, quantum discriminator, and an optimization rule. In particular,  the same loss function is employed to optimize $\bm{\theta}$ and $\bm{\gamma}$.    We  briefly explain the  main  differences to  the quantum patch GAN, with the details provided in the SM(E).

The main architecture  of     quantum batch GAN is illustrated in Fig.~\ref{fig:generator}(b). Both $G$ and $D$ are constructed using PQCs.  In the training procedure,  we first adopt a pertained oracle  to generate the latent state $\ket{\bm{z}^{(k)}}$, i.e., $\ket{\bm{z}^{(k)}} = 2^{-N_I}  \sum_i \ket{i}_I\ket{\bm{z}_i^{(k)}}_F$. Note that $N_F$ qubits, as in the first proposal does,  are  decomposed into two parts, where the first $N_G$ qubits are used to generate feature vectors and the remaining  $N_A$ are used to introduce nonlinearity. We then  apply the quantum generator $U_G(\bm{\theta}^{(k)})$ with $U_G(\bm{\theta}^{(k)}) \in \mathbb{C}^{2^{N_F}\times 2^{N_F}}$ to the latent state, where the generated state is $(\mathbb{I}_{2^N}\otimes U_G(\bm{\theta}^{(k)}) )\ket{ \bm{z}}$.  We then  apply the discriminator $U_D(\bm{\gamma}^{(k)}) \in \mathbb{C}^{2^{N_F}\times 2^{N_F}}$ to the generated state, i.e.,  \begin{equation}\ket{\Psi_t^{(k)}(\bm{z})}=(\mathbb{I}_{2^N}\otimes (U_D(\bm{\gamma}^{(k)})U_G(\bm{\theta}^{(k)}) ))\ket{\bm{z}^{(k)}}~.\end{equation} Finally, we  employ  a  \textit{Positive Operator Value Measurements} (POVM)  $\Pi$
to obtain the output of the discriminator $D(G(\bm{z}))$, i.e., $D(G(\bm{z}))   = \Tr(\Pi\ket{\Psi_t^{(k)}(\bm{z})}\bra{\Psi_t^{(k)}(\bm{z})})$ with $\Pi=\mathbb{I}_{2^{N-1}}\otimes \ket{0}\bra{0}$. Similarly,  we have $D( \bm{x})   = \bra{\Psi_t^{(k)}(\bm{x})}\Pi\ket{\Psi_t^{(k)}(\bm{x})}$ with $\ket{\Psi_t^{(k)}(\bm{x})} = (\mathbb{I}_{2^N}\otimes U_D(\bm{\gamma}^{(k)}) ) \ket{\bm{x}^{(k)}}$.

\section{SM (C): The implementation of quantum patch GAN}
The   quantum patch  GAN under the setting $N< \lceil \log M \rceil$ is composed of a quantum generator, a classical discriminator, and a classical optimizer. Here we  separately explain the implementation    of these three components.

\subsection{Quantum generator  }
Recall  that the same construction rule is applied to all $T$ sub-generators. Here we mainly exemplify $t$-th sub-generator $G_t$.  Quantum sub-generator $G_t$, analogous to classical generator, receives the input latent state $\ket{\bm{z}}$ and outputs the generated result $G_t(\bm{z})$.   We first introduce the preparation of the latent state $\ket{\bm{z}}$. We then describe the construction rule  of the  computation model $U_G(\bm{\theta})$ used in $G_t$. We last illustrate how to transform the generated quantum state to the generated example $G_t(\bm{z})$.

\textbf{Input latent state.} As explained in the main text, the latent state is prepared by applying a set of  rotation single qubit gates $U_S(\bm{\alpha}_z(i))$  to the input state $\ket{0}^{\otimes N}$ with $\bm{\alpha}_z\in \mathbb{R}^N$ and $U_S(\bm{\alpha}_z(i)) \in \{\Rx, \Ry,\Rz\}$, i.e.,  $\ket{\bm{z}}=(\bigotimes_{i=1}^{N}U_S(\bm{\alpha}_z(i)))\ket{0}^{\otimes N}$.
The exploitation of the latent variable input state $\ket{\bm{z}}$, analogous to  classical GAN, ensures quantum  GAN to be a probabilistic generative model \cite{bishop2006pattern}.    In the training procedure, the same $\ket{\bm{z}}$ is employed to input into all $T$ sub-generators. Such an operation guarantees that quantum patch  GAN is capable of converging to Nash equilibrium, as classical GAN claimed. The technical proof is shown in Section SM (D).

 \textbf{Computation model $U_{G_t}(\bm{\theta})$.}  The computation model $U_{G_t}(\bm{\theta})$ aims to map the input state $\ket{\bm{z}}$ to a specific quantum state that   well approximates  the target data.    Two key elements of our computation model are MPQC formulated in Section SM (A) and the  nonlinear transformation. The motivation to use  MPQC comes from two aspects. First, the structure of MPQC can be flexibly modified  to adapt the limitations of quantum  hardware, e.g., the restricted circuit depth and the allowable  number of  quantum gates \cite{preskill2018quantum}. Second, MPQC possesses a strong expressive power over classical circuit, which may contribute to quantum GANs to estimate the real data distribution \cite{du2018expressive}.  The adoption of the  nonlinear transformation intends to close the gap between the intrinsic mechanism of quantum computation  and the required setting for  generative models. Specifically, generative model essentially tries to learn a nonlinear map that transforms the  distribution $\Pro(\bm{z})$ to the target data distribution $\Pro_{data}(\bm{x})$. The intrinsic property of quantum computation implies that the trainable unitary, e.g., MPQC, can only linearly transform the input state to the output state.  Consequently, a nonlinear transformation strategy is demanded   for the quantum generator.

Here we  introduce one efficient method that enables $G_t$ to achieve  the nonlinear map. The central idea is adding an ancillary subsystem  in $G_t$ and then  tracing it out. Similar ideas have been broadly used in quantum discriminative models \cite{farhi2018classification,grant2018hierarchical,wan2017quantum}. Supposed that $G_t$ is an $N$ qubits system, we decompose it into the ancillary subsystem $\mathcal{A}$ with  $N_A$ qubits and the data subsystem with $N-N_A$ qubits.  We define the input state    $\ket{\bm{z}}$ as the following form, i.e.,
\begin{equation}\label{eqn:lat_z}
	\ket{\bm{z}} = \left(\bigotimes_{i=1, i\in S}^{N_S} \Ry(\bm{\alpha}_z{(i)})\bigotimes_{k=1, k\in [N]\setminus S}^{N-N_S} \text{I}_k\right)\ket{0}^{\otimes N}~,
\end{equation}
where   $S$ is the index set  with $S\subset [N]$ and $|S|=N_S$,  $\Ry(\bm{\alpha}_z{(i)})$ applies to the $i$-th qubit,  identity gate $\text{I}_k$ applies to the $k$-th qubit, and $\bm{\alpha}_z{(i)}$ refers to the $i$-th entry of the vector $\bm{\alpha}\in\mathbb{R}^{N_S}$ with $\bm{\alpha}$  being sampled from a predefined distribution.  We denote  MPQC   as the giant unitary $U_{G_t}(\bm{\theta})$ with $U_{G_t}(\bm{\theta})\in\mathbb{C}^{2^{N}\times 2^{N}}$.  The generated state  $\ket{\Psi_t(\bm{z})}$ for $G_t$ after interacting $U_t(\bm{\theta})$ with $\ket{{\bm{z}}}$ is
\begin{equation}
	\ket{\Psi_t(\bm{z})} = U_{G_t}(\bm{\theta}) \ket{{\bm{z}}}~.
\end{equation}
We then take the  partial measurement $\Pi_{\mathcal{A}}$ on  the ancillary subsystem $\mathcal{A}$ of $\ket{\Psi(\bm{z})}$, i.e., the post-measurement  quantum state $\rho_t(\bm{z})$ is
\begin{equation}\label{eqn:post_gene}
	\rho_t(\bm{z})  = \frac{\Tr_{\mathcal{A}}(\Pi_{\mathcal{A}}\ket{\Psi_t(\bm{z})}\bra{\Psi_t(\bm{z})})}{\Tr(\Pi_{\mathcal{A}} \otimes \mathbb{I}_{2^{N-N_A}} \ket{\Psi_t(\bm{z})}\bra{\Psi_t(\bm{z})})} ~.
\end{equation}
An immediate observation is that state $\rho_t(\bm{z})$ is a nonlinear map for $\ket{\bm{z}}$, since both the nominator and denominator of Eqn.~(\ref{eqn:post_gene}) are the function of the variable $\ket{\bm{z}}$.

\textbf{Output.} The output of  $G_t$, denoted as $G_t(\bm{z})$, is obtained by measuring $\rho_t(\bm{z})$ using a complete set of computation bases $\{\ket{j}\}_{j=0}^{2^{(N-N_A)}-1}$. For image generation, the measured result $\Pro(j)$ of the computation basis $\ket{j}$ represents the $j$-th pixel value for the $t$-th sub-generator, i.e.,
\begin{align}
	\Pro(J=j) = \Tr(\ket{j}\bra{j}\rho_t(\bm{z}))~.
\end{align}
Consequently,  we have $G_t(\bm{z})$
\begin{align}
	G_t(\bm{z}) = [\Pro(J=0), ..., \Pro(J=2^{(N-N_A)}-1)]~,
\end{align}
and the output for the generator $G(\bm{z})$ is
\begin{align}
	G(\bm{z}) = [G_1(\bm{z}), ..., G_T(\bm{z})]~.
\end{align}

\textit{\underline{Remark.}}  Other advanced  nonlinear mapping methods can be seamlessly  embedded into our  quantum generator. For example, it is feasible to employ classical activation function $f(\cdot)$, e.g., sigmoid function, to the generated result $G(\bm{z})$. It is intrigued to explore what kind of nonlinear mapping will lead to a better performance for our quantum GAN scheme.

\subsection{Discriminator}
The  discriminator  $D$  is constructed by employing classical neural networks, i.e., FCNN. The input of the discriminator  is the training data $\bm{x}$ or generated data $G(\bm{z})$. The output of $D$ is a scalar in the range between $0$ and $1$, i.e., $D(\bm{x}), D(G(\bm{z}))\in [0, 1]$. Recall that we label the training data as $1$ (True) and the generated data as $0$ (False).  The output of the discriminator can be treated as the confidence about the input data to be true or false. The ReLU mapping  is employed to build FCNN.  We customize the depth of the hidden layers  and the number of neurons in each layer  for different tasks.

\subsection{Loss function and optimization rule}
We modify the loss function  defined in Eqn.~(\ref{eqn:loss}) to train quantum patch GAN, i.e.,
		\begin{equation}\label{eqn:loss_QGAN1}
\min_{\bm{\theta}}\max_{\bm{\gamma}} \mathcal{L}(D_{\bm{\gamma}}(G_{\bm{\theta}}(\bm{z})),D_{\bm{\gamma}}(\bm{x})):=  \mathbb{E}_{\bm{x}\sim \Pro_{data}(\bm{x})} [\log D_{\bm{\gamma}}( \bm{x})] +\mathbb{E}_{\bm{z}\sim \Pro(\bm{z})} [\log(1-D_{\bm{\gamma}}(G_{\bm{\theta}}(\bm{z}))]~,
\end{equation}
 where the modified part is  setting   $G_{\bm{\theta}}(\bm{z}) = [G_{\bm{\theta},1}(\bm{z}), ...,G_{\bm{\theta},T}(\bm{z})]$.  In the training process, we optimize the trainable parameters $\bm{\theta}$ and $\bm{\gamma}$ iteratively, which is analogous to classical GAN. Especially,  we leverage a zeroth-order method \cite{mitarai2018quantum} and an  automatic differentiation package of PyTorch \cite{paszke2017automatic}  to optimize trainable parameters for the quantum generator and classical  discriminator, respectively. In particular,  to optimize the  classical discriminator $D$, we fix parameters $\bm{\theta}$ and use back-propagation to update the parameters $\bm{\gamma}$ according to  the obtained loss  \cite{goodfellow2016deep}. To optimize the  quantum generator $G$,  we keep the parameters $\bm{\gamma}$ fixed and employ the parameter shift rule   \cite{mitarai2018quantum} to compute the  gradients of PQC in a way that is compatible with   back-propagation. Denote $N_G$ and $N_D$ be the number of parameters for $G$ and $D$, i.e., $N_G=|\bm{\theta}|$ and $N_D=|\bm{\gamma}|$. The derivative of the $i$-th parameter $\bm{\theta}(i)$ with $i\in[N_G]$ can be computed by evaluating the original expectation twice, but with shifting $\bm{\theta}(i)$ to $\bm{\theta}(i)+\pi/2$ and $\bm{\theta}(i)-\pi/2$. In particular, we have
 \begin{align}\label{eqn:part_der}
\frac{ 	\partial \mathcal{L}(\bm{\theta},\bm{\gamma})}{\partial \bm{\theta}(i) } = \frac{\mathcal{L}(\bm{\theta}(1),...,\bm{\theta}(i)+\pi/2,...,\bm{\theta}(N_G),\bm{\gamma} ) -\mathcal{L}(\bm{\theta}(1),...,\bm{\theta}(i)-\pi/2,...,\bm{\theta}(N_G), \bm{\gamma}) }{2}~.
 \end{align}
 The update rule for $	\bm{\theta}$ at $k$-th iteration is
\begin{equation}
	\bm{\theta}^{(k)} = \bm{\theta}^{(k-1)} - \eta_G\frac{ 	\partial \mathcal{L}(\bm{\theta}^{(k-1)},\bm{\gamma}^{(k-1)})}{\partial \bm{\theta}^{(k-1)} }~,
\end{equation}
where $\eta_G$ is the learning rate. Analogous to the classical GAN, we iteratively  update parameters  $\bm{\theta}$ and  $\bm{\gamma}$ in total $K$ iterations.

\section{SM (D): The convergence guarantee of Quantum patch  GAN}\label{sec:NE}
Recall that classical GAN employs the following Lemma to prove its convergence.
\begin{lem}[Proposition 1, \cite{goodfellow2014generative}]\label{lem:1}
	For classical GAN, when $G$ is fixed, the optimal discriminator $D$ is
	 $$D^*(\bm{x})=\frac{\Pro_{data}(\bm{x})}{\Pro_{g}(\bm{x})+\Pro_{data}(\bm{x})}~.$$	
\end{lem}

In favor of Lemma \ref{lem:1}, the convergence property of classical GAN is summarized by the following two lemmas.
\begin{lem}[Theorem 1, \cite{goodfellow2014generative}]
Denote $C(G)$ as
$C(G):=\max_D \mathcal{L}(G, D),$ with ${L}(G, D)$ being loss function.
 The global minimum of the virtual training criterion C(G) is achieved if and only if $\Pro_g = \Pro_{data}$. At that point, $C(G)$ achieves the value  $-\log 4$.	
\end{lem}

\begin{lem}[Proposition 2,\cite{goodfellow2014generative}]\label{lem:2}
Denote $C(D)$ as
$C(D):=\min_G \mathcal{L}(G, D),$ with ${L}(G, D)$ being loss function.
	If the generator $G$ and discriminator $D$ have enough capacity, and at each iteration of GAN,  the discriminator is allowed to reach its optimum given $G$, and the generated distribution $\Pro_g$ is updated so as to improve the criterion $C(D)$
then $\Pro_g(\bm{x})$ converges to $\Pro_{data}(\bm{x})$. 	
\end{lem}

We now prove that the   quantum  patch   GAN possesses the identical convergence property as classical GAN does. Let $\Pro(\bm{z})$ be the distribution of the latent variable $\bm{z}$. We denote the probability distribution of generated images as $\Pro_g(\bm{x})$ with $\Pro_g(\bm{x})=\Pro_g(G(\bm{z}))$ and $ G(\bm{z}) = [G_1(\bm{z}), G_2(\bm{z}), ...,G_T(\bm{z})] ~.$

\begin{thm}\label{thm:1}
	In quantum  patch GAN,	for $G$ fixed, the optimal discriminator $D$ is
	 $$D^*(\bm{x})=\frac{\Pro_{data}(\bm{x})}{\Pro_{g}(\bm{x})+\Pro_{data}(\bm{x})}~.$$
\end{thm}

\begin{proof}
Given fixed generator $G$, we formulate the relation between $\Pro_g(\bm{x})$ and $\Pro(\bm{z})$ as follows.
	\begin{align}\label{eqn:thm1_1}
	\Pro_g(\bm{x}) =  \int_{\{[G_1(\bm{z}),G_2(\bm{z}),...,G_T(\bm{z})]=\bm{x}\}} \Pro({\bm{z}})d\bm{z} = \int_{\{G(\bm{z})=x\}} \Pro({\bm{z}})d\bm{z}~,
\end{align}

We then expand the loss function of quantum patch GAN and obtain
\begin{align}\label{eqn:thm1_2}
	& 	\mathbb{E}_{\bm{x}\sim \Pro_{data}}[\log(D(\bm{x}))]+ \mathbb{E}_{\bm{z}\sim \Pro(\bm{z})}[\log(1-D(G(\bm{z})))]  \nonumber\\
	= & \int_{\bm{x}}\Pro_{data}(\bm{x})\log(D(\bm{x})) d\bm{x} +  \int_{\bm{z}}  \Pro({\bm{z}})\log(1-D(G(\bm{z})))d\bm{z} ~.
\end{align}

In conjunction with Eqn.~(\ref{eqn:thm1_1}) and Eqn.~(\ref{eqn:thm1_2}), we have
\begin{align}
& \mathbb{E}_{\bm{x}\sim \Pro_{data}}[\log(D(\bm{x}))]+ \mathbb{E}_{\bm{z}\sim \Pro(\bm{z})}[\log(1-D(G(\bm{z})))]  \nonumber\\
= &  \int_{\bm{x}}\Pro_{data}(\bm{x})\log(D(\bm{x})) d\bm{x} + \int_{\bm{z}} \Pro({\bm{z}})\log(1-D(G(\bm{z})))d\bm{z} \nonumber\\
 = &\int_{\bm{x}}\Pro_{data}(\bm{x})\log(D(\bm{x})) d\bm{x} +  \int_{\{G(\bm{z})=\bm{x}\}} \Pro({\bm{z}}) \log(1-D(G(\bm{z}))) d\bm{z} d\bm{x} \nonumber \\
 = & \int_{\bm{x}}\Pro_{data}(\bm{x})\log(D(\bm{x})) d\bm{x} + \int_{\bm{x}}\log(1-D(\bm{x})) d\bm{x} \int_{\{G(\bm{z})=\bm{x}\}} \Pro({\bm{z}})  d\bm{z} \nonumber \\
   = & \int_{\bm{x}}\Pro_{data}(\bm{x})\log(D(\bm{x}))  + \Pro_g(\bm{x})\log(1-D(\bm{x})) d\bm{x}
\end{align}

Since both $\Pro_g(\bm{x})$ and $\Pro_{data}(\bm{x})$ are fixed, the minimum of the above equation is  $$D^*(\bm{x})=\frac{\Pro_{data}(\bm{x})}{\Pro                                                       _{g}(\bm{x})+\Pro_{data}(\bm{x})}~.$$							
\end{proof}

An immediate observation of Theorem \ref{thm:1} is that
\begin{coro}
For quantum   GAN,	 the global minimum is achieved if and only if $\Pro_g = \Pro_{data}$. If the generator $G$ and discriminator $D$ have enough capacity, and the discriminator can reach its optimum given $G$ at each iteration, and the generated distribution $\Pro_g$ is updated so as to improve the criterion $C(D)$  then $\Pro_g(\bm{x})$ converges to $\Pro_{data}(\bm{x})$. 	
\end{coro}
\begin{proof}
	The same optimal discriminator (as indicated by Lemma \ref{lem:1} and Theorem \ref{thm:1}), loss function, and updating rule imply that the convergence results obtained by classical GAN are also satisfied to quantum  patch GAN.
\end{proof}

\section{SM (E): The implementation of quantum batch GAN  }
The   proposed  quantum batch   GAN under the setting  $N > \lceil \log M \rceil$ employs a quantum generator and discriminator to play a minimax game.  Given an $N$-qubits quantum system, we divide $N$-qubits into the index register $\mathcal{R}_I$ with $N_I$ qubits and the feature register $\mathcal{R}_F$ with $N_F$ qubits, i.e., $N=N_I+N_F$. The feature register $\mathcal{R}_F$ can be further partitioned into three parts, i.e., $N_D$ qubits are used to generate fake examples, $N_{A_G}$ qubits are used to conduct nonlinear operations for $G$, and  $N_{A_D}$ qubits are used to conduct nonlinear operations for $D$ with $N_F = N_D +N_{A_G} +N_{A_D}$.
  Such a decomposition  enables us to effectively acquire the mini-batch gradient information by simple measurements. Considering that the mechanism of the quantum batch GAN is in the same vein with the quantum patch GAN, here we mainly concentrate on the distinguished techniques used in the quantum batch GAN.

\textbf{Input state.}  To capture the mini-batch gradient information, we employ two oracles  $U_{\bm{z}}$ and $U_{\bm{x}}$ to encode  different latent  vectors   and classical training examples into quantum states, respectively. Following the same notations used in the main text, we denote the mini-batch size    as $|B_k|=2^{N_I}$. For  $U_{\bm{z}}$, we have  $U_{\bm{z}}:\ket{0}^{\otimes N_I}_I\ket{0}^{\otimes N_F}_F\rightarrow 2^{-N_I}\sum_i \ket{i}_I\ket{\bm{z}^{(i)}}_F$. With a slight abuse of notation,  $\bm{z}^{(i)}$ refers to $i$-th latent vector and $\ket{\bm{z}^{(i)}}=\ket{\bar{\bm{z}}^{(i)}}\ket{0}^{\otimes N_{A_D}}$, where $\ket{\bar{\bm{z}}^{(i)}}\in \mathbb{C}^{2^{N_I+N_{A_G}}}$ follows the same form defined  in  Eqn.~(\ref{eqn:lat_z}).  Similarly, for $U_{\bm{x}}$, we have  $U_{\bm{x}}:\ket{0}_I\ket{0}_F\rightarrow 2^{-N_I}\sum_i \ket{i}_I\ket{\bm{x}^{(i)}}_F$. For a dataset of $2^{N_I}$ inputs with $M$ features, the complexity of encoding a full data set by the quantum system using amplitude encoding method is $O({2^{{N_I}}}M/({N_I}{\rm{log(}}M{\rm{)))}}$ \cite{schuld2018supervised, knill1995approximation, mottonen2004quantum, vartiainen2004efficient, plesch2011quantum}. Thus, for data encoding, the runtime of state preparation for quantum machine learning using amplitude encoding is basically consistent with that of classic machine learning, since the encoding complexity of classic machine learning is at least $O({2^{{N_I}}}M)$. However, the number of qubits required for quantum machine learning is $N_Ilog(M)$, while classical quantum machine learning requires at least $O(2^{N_I}M)$ bits.

An accurate construction of   $U_{\bm{z}}$ and $U_{\bm{x}}$  requires numerous multi-controlled quantum gates, which is inhospitable to near-term quantum devices. To overcome this issue, an effective way is to employ the pertained oracles that approximate $U_{\bm{z}}$ and $U_{\bm{x}}$ to accomplish the learning tasks. Such a pre-training method have been broadly  investigated \cite{liu2018differentiable,huggins2018towards,benedetti2019generative,zoufal2019quantum,romero2019variational}.

\textbf{Computation model   $U_G(\bm{\theta})$.} The quantum generator $G$ is built by MPQC $U_G(\bm{\theta})$ associated with the nonlinear mappings.  As illustrated in the main text, $U_G(\bm{\theta})\in\mathbb{C}^{2^{N_D+N_{A_G}}\times 2^{N_D+N_{A_G}}}$ only operates with the feature register $\mathcal{R}_F$. In particular, to generate fake data, we first apply $\mathbb{I}_{2^{N_I}}\otimes U_G(\bm{\theta})\otimes \mathbb{I}_{2^{N_{A_D}}}$ to the input state, i.e., $\ket{\Psi(\bm{z} )} = 2^{-N_I}\sum_i \ket{i}_I U_G\otimes \mathbb{I}_{2^{N_{A_D}}}(\bm{\theta})\ket{\bm{z}^{(i)}}$. We then take a partial measurement $\Pi_{\mathcal{A}_G}$ as defined in Eqn.~(\ref{eqn:post_gene}), e.g., $\Pi_{\mathcal{A}_G}=(\ket{0}\bra{0})^{\otimes N_{A_G}}$,  to introduce the nonlinearity. The generated state $\ket{G(\bm{z})}$ corresponding to $|B_k|$ fake examples is
\begin{equation}\label{eqn:GZ}
\ket{G(\bm{z})} := 2^{-N_I}\sum_i\ket{i}_I \ket{G(\bm{z}^{(i)})}_F=    \frac{\mathbb{I}_{2^{N_I}}\otimes \Pi_{\mathcal{A}_G}\otimes\mathbb{I}_{2^{N_D+N_{A_D}}} \ket{\Psi(\bm{z} )}}{\sqrt{\Tr(\mathbb{I}_{2^{N_I}}\otimes \Pi_{\mathcal{A}_G} \otimes \mathbb{I}_{2^{N_D+N_{A_D}}}\ket{\Psi(\bm{z})} \bra{\Psi(\bm{z})} )} }	~.
\end{equation}

In the training procedure, we directly apply the quantum discriminator to operate with the generated state $\ket{G(\bm{z})}$. In the image generation stage, we employ POVM to measure the state $\ket{G(\bm{z})}$, i.e., the $i$-th image $G(\bm{z}^{(i)})$ with $i\in B_k$ is $G(\bm{z}^{(i)})=[\Pro(J=0|I=i),...,\Pro(J=2^{N_D}-1|I=i)]$ with $\Pro(J=j|I=i) = \Tr(\ket{i}_I\ket{j}_F\bra{i}_I\bra{j}_F \ket{G(\bm{z})}\bra{G(\bm{z})})$.

\textbf{Computation model   $U_D(\bm{\gamma})$.}
Quantum discriminator $D$, implemented by MPQC  $U_D(\bm{\gamma})$ associated with the nonlinear operations, aims to output a scalar that represents to the averaged classification accuracy. Given a state $\ket{\bm{x}}$ that represents $|B_k|$ real examples, we first apply $\mathbb{I}_{2^{N_I+N_{A_G}}}\otimes U_D(\bm{\gamma})$ to the state $\ket{\bm{x}}$, i.e.,   $\ket{\Phi(\bm{x})} = 2^{-N_I}\sum_i \ket{i}_I \mathbb{I}_{2^{N_{A_G}}}\otimes U_D(\bm{\gamma})\ket{\bm{x}^{(i)}}_F $. We then use a partial measurement $\Pi_{\mathcal{A}_D}$, e.g.,  $\Pi_{\mathcal{A}_D}=(\ket{0}\bra{0})^{\otimes N_{A_D}}$, to introduce the nonlinearity. The generated state $\ket{D(\bm{x})}$ corresponding to the classification result for $|B_k|$ examples is
\begin{equation}\label{eqn:mini_q}
	\ket{D(\bm{x})} := 2^{-N_I}\sum_i \ket{i}_I\ket{D(\bm{x}^{(i)})}_F = \frac{\mathbb{I}_{2^{N-N_{A_D}}}\otimes \Pi_{\mathcal{A}_D} \ket{\Psi(\bm{x} )}}{\sqrt{\Tr(\mathbb{I}_{2^{N-N_{A_D}}}\otimes \Pi_{\mathcal{A}_D} \ket{\Psi(\bm{x})} \bra{\Psi(\bm{x})} )} }	~.
\end{equation}
Similarly, given  the state $\ket{G(\bm{z})}$ in Eqn.~(\ref{eqn:GZ}) that represents $|B_k|$ fake examples, we adopt the same method to obtain the state $\ket{D(G(\bm{z}))}$ with  $\ket{D(G(\bm{z}))} =  2^{-N_I}\sum_i \ket{i}_I\ket{D(G(\bm{z}^{(i)}))}_F$.  For each example $\bm{x}^{(i)}$  or $G(\bm{z}^{(i)})$,  the classification accuracy $D(\bm{x}^{(i)})$  or $D(G(\bm{z}^{(i)})$ is obtained by applying  POVM $\Pi_{o}=\mathbb{I}_{2^{N_D+N_{A_D}}}$ on $\ket{D(\bm{x}^{(i)})}$ or  $\ket{D(G(\bm{z}^{(i)}))}$, i.e.,  $D(\bm{x}^{(i)})=\Tr(\Pi_o \ket{D(\bm{x}^{(i)})}\bra{D(\bm{x}^{(i)})})$ or $D(G(\bm{z}^{(i)}))=\Tr(\Pi_o \ket{D(G(\bm{z}^{(i)}))}\bra{D(G(\bm{z}^{(i)}))})$. As formulated in Eqn.~(\ref{eqn:mini_q}), the averaged classification accuracy $2^{-N_I}D(\bm{x})$ is acquired by applying POVM $\Pi_{o}=\mathbb{I}_{2^{N-1}}\ket{0}\bra{0}$ to $\ket{D(\bm{x})}$, i.e., $2^{-N_I}D(\bm{x}) = \Tr(\Pi_{o}\ket{D(\bm{x})}\bra{D(\bm{x})})$. Likewise, the averaged classification accuracy for the generated examples $2^{-N_I}D(G(\bm{z}))$ is acquired by applying   $\Pi_{o}$ to $\ket{D(G(\bm{z}))}$, i.e.,  $2^{-N_I}D(G(\bm{z})) = \Tr(\Pi_{o}\ket{D(G(\bm{z}))}\bra{D(G(\bm{z}))})$.   In conjunction with Eqn.~(\ref{eqn:mini_batch}) and Eqn.~(\ref{eqn:loss_QGAN1}), the mini-batch gradient information can be effectively acquired by taking $\Pi_o$ on two states $\ket{D(G(\bm{z}))}$ and $\ket{D(\bm{x})}$.

\textit{\underline{Remark.}}

1). It is noteworthy that, by introducing $N_I$ additional qubits to encode $2^{N_I}$ inputs as a superposition state, quantum batch GAN could obtain the batch gradient descent of all the $2^{N_I}$ inputs in one training process. This shows that quantum batch GAN has the potential to efficiently process big data.

2). Analogous to binary classification task, we need to measure output state to acquire the information about if the input is `fake' or `real'. Since quantum batch GAN employ the quantum discriminator, theoretically, measuring one qubit is enough to distinguish between `real' and `fake' images, and then obtain the gradient information. Therefore, the number of measurements required for quantum batch GAN during training procedure is quite small, and theoretically will not increase with the size of the system. For example, the statistical error of 10,000 measurements on a qubit is about 0.01, which is basically enough for the training procedure in most cases. In addition, as discussed in \cite{2019arXiv191001155S}, finite number of  measurements could lead to a unbiased estimators for the gradient, which effectively avoids the saddle points and possesses the convergence guarantees.

\section{SM (F): Experiment Details}
	
In this section, we first specify the parameter settings of the exploited superconducting quantum processor.  We next provide the experiment  details for the hand-written digit image generation task. We last  demonstrate the experiment  details for  the gray-scale bar image generation.

\subsection{Superconducting quantum processor}

In all experiments, the six qubits (see Fig.~\ref{fig:set-up}) are chosen from a 12-qubit superconducting quantum processor. The processor has qubits lying on a 1D chain, and the qubits are capacitively coupled to their nearest neighbors (the coupling strength is about 12 MHz). Each qubit has a microwave drive line (XY), a fast flux-bias line (Z) and a readout resonator. All readout resonators are coupled to a common transmission line for state readout. The single-qubit rotation gates are implemented by driving the XY control lines, and the average gate fidelity of single-qubit gates is about 0.9994. The controlled-Z (CZ) gate is implemented by driving the Z line using the ``fast  adiabatic’’ method, whose average gate fidelity is about $0.985$. During the experiments, we only calibrated qubit readouts every hour but did not calibrate the quantum gate operations, even over four days of training. Thus, the optimization of our quantum GAN scheme is very robust to noise. The performances of the six qubits we chosen in our experiment are listed in Table.~\ref{tab:QPU}.

 \begin{figure*}[htp]
			\centering   \includegraphics[width=0.65\textwidth]{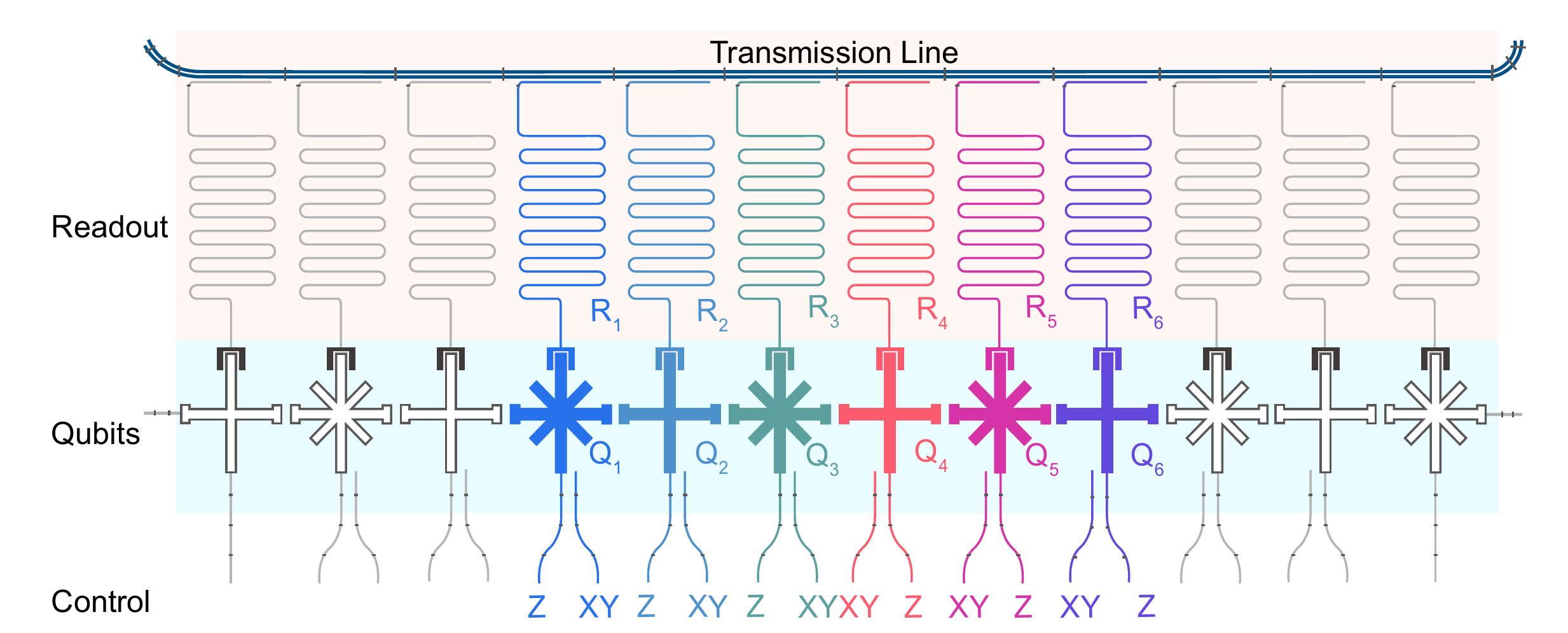}
\caption{\small{\textbf{Experiment set-up.} There are 12 qubits in total in our superconducting quantum processor, from which we choose six adjacent qubits labelled with $Q_1$ to $Q_6$ to perform the experiment. Each qubit couples to a corresponding resonator for state readout. For each qubit, individual capacitively-coupled microwave control lines (XY) and inductively-coupled bias lines (Z) enable full control of qubit operations.} }
			\label{fig:set-up}
\end{figure*}


 \begin{table}[h!]
\centering
\begin{tabular}{l l l l l l l l }
\hline \hline
Qubit & Q1 & Q2 & Q3 & Q4 & Q5 & Q6& AVG \\ \hline
$\omega_{10}/2\pi$ (GHz) & 4.210 & 5.006 & 4.141 & 5.046 & 4.226 & 5.132 &\multicolumn{1}{c}{-} \\
$T_1$ ($\mu s$) & 37.2 & 34.5 & 35.1 & 30.1 & 39.4 & 36.3 & 35.4 \\
$T_2^*$ ($\mu s$) & 2.6 & 4.8 & 1.5 & 8.6 & 2.4 & 5.4 & 4.2 \\
$f_{00}$ & 0.947 & 0.955 & 0.959 & 0.982 & 0.962 & 0.981 & 0.964 \\
$f_{11}$ & 0.873 & 0.913 & 0.889 & 0.919 & 0.904 & 0.93 & 0.905 \\ \hline
X/2 gate fidelity & 0.9993 & 0.9993 & 0.9992 & 0.9995 & 0.9993 & 0.9996& 0.9994 \\
CZ gate fidelity & \multicolumn{7}{l}{~~~~~0.987  ~~    0.985  ~    0.986    ~~ 0.972   ~0.994     ~~~~ 0.985} \\ \hline \hline
\end{tabular}
 \caption{\small{\textbf{Performance of qubits}. $\omega_{10}$ is idle points of qubits. $T_1$ and $T_2^*$ are the energy relaxation time and dephasing time, respectively. $f_{00}$ ($f_{11}$) is the possibility of correctly readout of qubit state in $\ket{0}$ ($\ket{1}$) after successfully initialized in $\ket{0}$ ($\ket{1}$) state. X/2 gate fidelity and CZ gate fidelity are single and two-qubit gate fidelities obtained via performing randomized benchmarking.       }}
   \label{tab:QPU}
\end{table}

\subsection{Hand-written digit image generation}
Here, we provide the hyper-parameter settings of quantum patch GAN in hand-written  digits `0' and `1' image generation tasks. In particular, we set $N_S = N = 5$ defined in Eqn.~(\ref{eqn:lat_z}) to generate latent states. The number of sub-generators and  layers for each $U_{G_t}$ are set as   $T=4$ and  $L=5$, respectively.   To compress the depth of the quantum circuits, we set all trainable single qubit gates as \Ry. Equivalently, the total number of trainable parameters for quantum generator  $G$ is in total $T\times L \times N =100$. The number of measurements to readout the quantum state is set as $3000$. Moreover, the employed   discriminator  used for quantum patch GANis implemented by FCNN with two hidden layers, and the number of hidden neurons for the first and second hidden layer is $64$ and $16$, respectively. In the training  procedure, we set the learning rates as $\eta_G=0.05$ and $\eta_D=0.001$ for quantum patch GAN.   The number of measurements to estimate the partial derivation in Eqn.~(\ref{eqn:part_der}) is set as $3000$.

\subsubsection{Some discussion about the setting about the number of measurements}
Here we devise a numerical simulation to indicate that, for the hand-written image generation task that using quantum patch GAN with 5 qubits, $K=3000$ shots measurement is a good hyper-parameter to achieve the desired generative performance under a reasonable running time. Specifically, we employ the quantum generator used in quantum patch GAN to accomplish the discrete Gaussian distribution approximation task. Formally, the discrete Gaussian distribution $\pi(x;\mu,\sigma)$   is defined as
\begin{equation}
	\pi(x;\mu,\sigma) = \exp\left({-\frac{(x-\mu)^2}{2\sigma^2}}\right)/Z~,
\end{equation}
where $x\in[0,31]$  and $Z$ being the normalization factor.  The discrete Gaussian   $\pi(x;\mu,\sigma)$ can be effectively represented by the quantum state using five qubits. Let the target quantum state expressed by five qubits be  $\ket{\pi}$, where the outcome measured by the  computation basis $\ket{k}$  with $k\in[0,31]$ is $\exp({-\frac{(k-\mu)^2}{2\sigma^2}})/Z$.

We now exploit quantum generator used in quantum patch GAN approximate the target state $\ket{\pi}$, or equivalently, to learn the discrete Gaussian distribution $\pi(x;\mu,\sigma)$. Denote the generated state of the employed quantum generator as $\ket{\psi(\bm{\theta})}$,
\begin{equation}
	\ket{\psi(\bm{\theta})} = \prod_{i=1}^L U_{i}(\bm{\theta})\ket{0}^{\otimes 5}~,
\end{equation}
where $U_{i}(\bm{\theta})$ refers to PQC. The probability distribution formulated by $\ket{\psi(\bm{\theta})}$ is denoted as $q_{\bm{\theta}}$, i.e., $q(X=k)=|\braket{k|\psi(\bm{\theta})}|^2$. In the training procedure, we continuously update $\bm{\theta}$ to minimize the maximum mean discrepancy (MMD) $\mathcal{L}$ between two distributions $q(x)$ and $p(x)$, i.e.,
\begin{equation}
	\mathcal{L}(q(x), p(y)) = \mathbb{E}_{x\sim q(x), y\sim p(y)}[\mathbb{K}(x,y)] - 2\mathbb{E}_{x\sim q(x), y\sim \pi(y)}[\mathbb{K}(x,y)] + \mathbb{E}_{x\sim \pi(x), y\sim \pi(y)}[\mathbb{K}(x,y)]~,
\end{equation}
where $\mathbb{K}(x,y)=\frac{1}{c}\sum_{i=1}^c\exp(-|x-y|^2/(2\sigma_i^2))$ and $c\in\mathbb{N}$ is a hyper-parameter. At each iteration, we first readout the probability distribution $q_{\bm{\theta}}$  and compute the gradient  $\partial\mathcal{L}/\partial \bm{\theta}$.  The hyper-parameter setting is as follows. The total number of iteration is set as $T=800$. The learning rate is set as $lr=0.01$. The circuit depth $L$ is set as $L=5$. Figure \ref{fig:k3000} illustrates the simulation results. As shown in upper panel, with setting $K=3000$, the approximated Gaussian distribution can well match the target distribution. Moreover, the lower panel shows that, for the setting   and   (highlighted by red color), the training loss is continuously decreasing with the increased number of iterations. Celebrated by such a simulation result, we conclude that $K=3000$ is sufficient to acquire optimization information, which ensures the performance of quantum patch GAN.

\begin{figure*}[h]
	\centering
\includegraphics[width=0.86\textwidth]{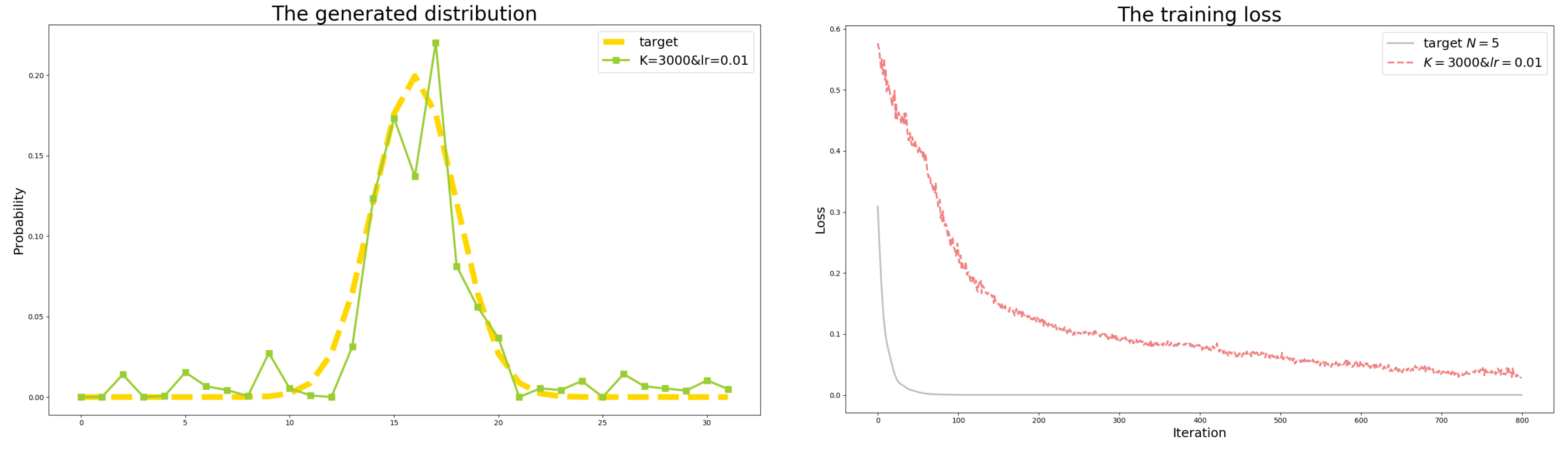}
\caption{\small{The simulation results for approximating discrete Gaussian distribution with finite measurements $K=3000$. The left panel shows the performance of approximated discrete Gaussian. The label `target' refers to the target Gaussian distribution to be approximated. The label `lr' refers to learning rate.  Similarly, the right panel illustrates the corresponding training loss. }}
\label{fig:k3000}
\end{figure*}

\subsection{Gray-scale bar image generation}
\subsubsection{The gray-scale bar image dataset}
Here, we first address the motivation of constructing the grey-scale bar dataset, and discuss the requirements that need to be considered for constructing such a dataset. The gray-scale bar image dataset is used to explore how the performance of quantum patch GAN and quantum batch GAN. To evaluate the performance of two quantum GANs, the employed dataset should satisfy the following two requirements:

1. Given a dataset, the preparation of quantum state that corresponds to the classical input, is required to be efficient, which only cost shallow or constant circuit depth.

2. The employed dataset $\mathcal{D}$ should be sampled from a continuous distribution, i.e., $\mathcal{D}\sim P_{data}(\bm{x})$.

The Requirement 1 origins from the practical limitation. Considering that the noise of quantum system is exponentially increased in terms of the circuit depth, it is unfavorable that encoding classical input into quantum states affects our analysis results. Equivalently, an  efficient method to prepare quantum input facilitates us to eliminate the effects of the encoding issue,  and enables us to better explore how the performance of quantum batch GAN. The Requirement 2 ensures that the employed dataset is sufficiently `complicated' to learn.

The construction rules for the gray-scale bar dataset  are  as follows.  Denote the training dataset as $\mathcal{D}=\{\bm{x}_i\}_{i=1}^{N_e}$ with $\mathcal{D}\sim \Pro_{data}(\bm{x})$, where $N_e$ is the number of  examples   and  $\bm{x}_i\in\mathbb{R}^M$ refers to the $i$-th example with feature dimension $M$.     Denote the    pixel value at the  $i$-th row and $j$-th column as $\bm{x}_{ij}$, a   valid  gray-scale bar image $\bm{x}\in \mathbb{R}^{m \times m}$ with  $M=m^2$  satisfies   $\bm{x}_{i0} \sim \text{unif}(0.4,0.6)$,   $\bm{x}_{i1} = 1- \bm{x}_{i0}$, and $\bm{x}_{ij}=0$, $\forall i \in [m] $ and $ \forall j \in [m]\setminus \{0,1\}$. In our experiment, we collect a training dataset with $N_e=1000$ examples for the case $m=2$.

The gray-scale bar dataset cleverly meets the two requirements nominated above, which motivates us to use it to investigate the performance of quantum GAN.  On the one hand, we can effectively encode the training data into quantum state by using one circuit depth  that is composed of $\textrm{RY}$ gates. For example, for the $2\times 2$ pixels setting, the image $\bm{x}=[0.45, 0, 0.55, 0]$, the corresponding quantum state can be generated by applying $\textrm{RY}(\gamma_1)\otimes \textrm{RY}(\gamma_2)$ to the initial state $\ket{00}$, where $\gamma_1=2*\arccos (\sqrt{0.45})$ and $\gamma_2=0$.  On the other hand, since the data distribution of gray-scale bar images is continuos, we can better evaluate if quantum GAN learns the real data distribution from finite training examples.

In our experiments, to evaluate the FD score, we sample 1000 generated examples after every 50 iterations, and usually we calculate the FD score of the generated examples after the training is completely over. In order to flexibly monitor the training procedure, we set a constraint to the gray-scale bar dataset that $x_{i0}\sim \textrm{unif}(0.4, 0.6)$. Instead of calculating FD score after training, we can check if the generated image satisfies such a constraint to roughly evaluate the performance of quantum generator at each iteration during the training procedure.

\subsubsection{Experimental details}
In the main text, we first apply the quantum patch GAN to generate $2 \times 2$ gray-scale bar images. Specifically, the hyper-parameters setting for quantum patch GAN is $N = 3$, $N_S=N$, and $L = 3$. In addition, we fix all $U_S$ to be RY, and the learning rates are set as $\eta_G = 0.05$ and $\eta_D = 0.001$. The number of measurements to readout the quantum state is set as 3000. The total number of trainable parameters for the quantum generator is 9 with $T = 1$.

We then apply the quantum batch GAN to generate gray-scale bar images. The hyper-parameters setting is identical to the quantum patch GAN, expect for the construction of discriminator. In particular, any quantum discriminative model based on amplitude-encoding method can be employed as the discriminator of quantum batch GAN. Here we utilize the quantum discriminator model proposed by \cite{schuld2019quantum} as our quantum discriminator. The total number of trainable parameters for the quantum discriminator is 12. Experiments demonstrate that the quantum batch GAN achieved reasonable generation performance, even though the quantum discriminator employed much fewer parameters than other configurations (The classical discriminator used in the classical GAN-MLP, classical GAN-CNN and quantum patch GAN has 96 parameters).

\begin{figure*}[h]
	\centering
\includegraphics[width=1.0\textwidth]{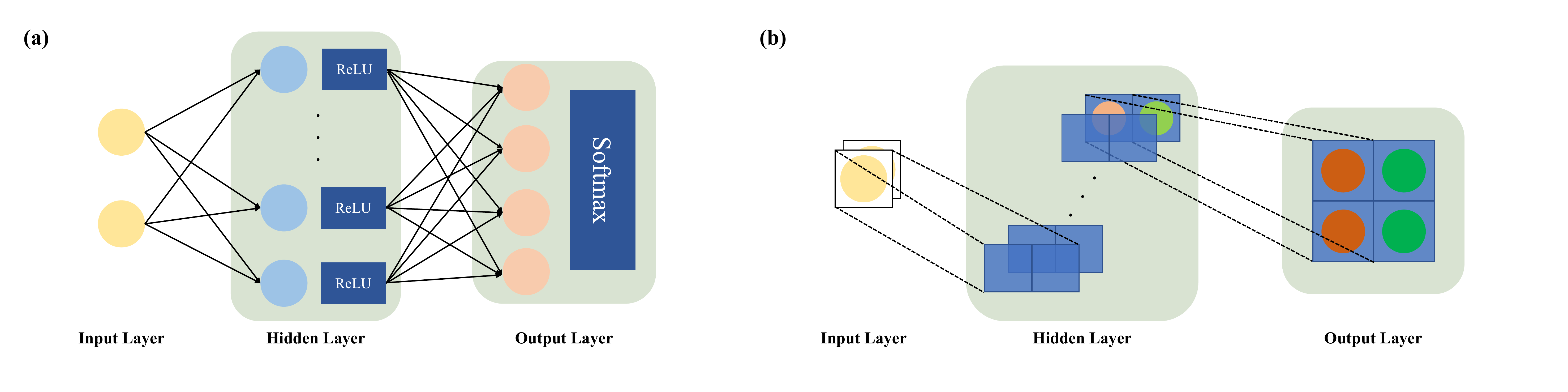}
\caption{\small{The architectures of employed two types of classical GANs. (a) the generator of GAN-MLP, a classical GAN model with multilayer perceptron generator. (b) the generator of GAN-CNN, a classical GAN model with convolutional generator.}}
\label{fig:GAN_model}
\end{figure*}

To better justify the capability and performance of both the quantum patch GAN and quantum batch GAN, we implemented two types of classical GANs as reference. Firstly, we built multilayer perceptron (MLP) generators with one hidden layer. As shown in Fig.~\ref{fig:GAN_model}(a), the input layer of MLP consists of one or two neurons, and noise sampled from the standard Gaussian distribution are feed as inputs. ReLU activations are added in the hidden layer to perform nonlinear transformation. In the output layer, the activation function, Softmax, is employed. It is mainly because that the Softmax activation share the same function with normalization constraint of the quantum generator, i.e., enforcing the sum of generator outputs to be equal to 1.
The exploited discriminator D has the identical configuration with quantum patch GAN. Moreover, following the implementation of the original GAN, the adversarial training process are formulated as,

\begin{equation}\label{eqn:loss_QGAN}
\begin{aligned}
&
\text{min}_D~\mathbb{E}_{\bm{x}\sim \Pro_{data}} [\log D( \bm{x})] +\mathbb{E}_{\bm{z}\sim \Pro(\bm{z})} [\log(1-D(G(\bm{z})))],~\\
&\text{max}_G~\mathbb{E}_{\bm{z}\sim \Pro_{p}(\bm{z})} [\log D(G(\bm{z}))].
\end{aligned}
\end{equation}

In the generator of GAN-CNN (Figure ~\ref{fig:GAN_model}(b)), the convolutional kernels with shape `(1$\times$2)' and `(2$\times$1)' are applied to the input noise and hidden features, respectively. Giving a sampled noised as input, the CNN generator can directly output a $2\times2$ gray-scale bar image. Similar to the MLP generator, nonlinear activations are added in the hidden and output layer. For both GAN-MLP and GAN-CNN, the stochastic gradient descent (SGD) \cite{ruder2016overview} is utilized to the classical generator and discriminator alternately.



To comprehensively explore the capability of classical GANs, grid-search is performed to find the optimal hyperparameters for each classical GAN model. Specifically, we start searching the learning rate from $10^{-4}$, and gradually increase it to $5\times 10^{-3}$ by $10^{-4}$ each step. For the coefficients of optimizers, such as Nesterov momentum of SGD, we start searching from 0.5, and increase them to 1 by 0.1 each step. To ensure classical GANs could achieve reasonable results, we trained each parameter combination 10 times, and save 5 models with higher FD scores.
				
\bibliographystyle{naturemag}	
\bibliography{myref2}